\documentclass[11pt]{llncs}
\usepackage[utf8]{inputenc}
\usepackage[T1]{fontenc}
\usepackage{bm}
\usepackage{amsmath,amsfonts,amssymb,multirow}

\usepackage{fullpage}
\usepackage{epsfig}
\usepackage{graphics}
\usepackage{float}

\usepackage[dvipsnames,usenames]{color}
\usepackage[colorlinks=true,urlcolor=Blue,citecolor=Green,linkcolor=BrickRed]{hyperref}
\usepackage[usenames,dvipsnames]{xcolor}
\usepackage{algorithm}
\usepackage{algorithmic}

\newcommand{\proc}{\textsc}
\newcommand{\set}[1]{\ensuremath{\{#1\}}}
\newcommand{\rvec}{\bm{r}}
\newcommand{\heit}{\text{\it height}}

\newcommand{\decreaseKey}{\text{\it decreaseKey}}
\newcommand{\increaseKey}{\text{\it increaseKey}}
\newcommand{\minItem}{\text{\it minItem}}
\newcommand{\minKey}{\text{\it minKey}}
\newcommand{\getKey}{\text{\it getKey}}
\newcommand{\entry}{\text{entry}}
\newcommand{\debt}{\text{debt}}
\newcommand{\upDebt}{\text{upDebt}}
\newcommand{\credit}{\text{credit}}
\newcommand{\parent}{\mbox{parent}}
\newcommand{\pstart}{\mbox{start}}
\newcommand{\pend}{\mbox{end}}

\newtheorem{observation}{Observation}

\makeatletter

  \newcommand{\oren}[1]{}
  \newcommand{\shay}[1]{}
  \newcommand{\yahav}[1]{}

\begin{document}

\title{Faster Shortest Paths in Dense
Distance Graphs,\\ with Applications\thanks{The research was supported in part by Israel Science Foundation grant 794/13.}}

\author{
Shay Mozes\inst{1}\and Yahav Nussbaum\inst{2} \and Oren Weimann\inst{3}
}
\institute{
IDC Herzliya, \href{mailto:smozes@idc.ac.il}{smozes@idc.ac.il}   \and
University of Haifa, \href{mailto:yahav.nussbaum@cs.tau.ac.il}{yahav.nussbaum@cs.tau.ac.il}   \and
University of Haifa, \href{mailto:oren@cs.haifa.ac.il}{oren@cs.haifa.ac.il}  
}

\date{}
\maketitle
\begin{abstract}

We show how to combine two techniques for efficiently computing shortest
paths in directed planar graphs. The first is the linear-time
shortest-path algorithm of Henzinger, Klein, Subramanian, and
Rao [STOC'94].
The second is Fakcharoenphol and Rao's algorithm [FOCS'01] for emulating
Dijkstra's algorithm on the {\em dense distance graph} (DDG).
A DDG is defined for a decomposition of a planar graph $G$
into regions of at most $r$ vertices each, for some parameter $r < n$.
The vertex set of the DDG is the set of $\Theta(n/\sqrt r)$ vertices
of $G$ that belong to more than one region (boundary vertices). The DDG has $\Theta(n)$ arcs, such that distances
in the DDG are equal to the distances in $G$.
Fakcharoenphol and Rao's implementation of Dijkstra's
algorithm on the DDG (nicknamed {\em FR-Dijkstra}) runs in 
 $O(n\log(n) r^{-1/2} \log r)$
time, and is a key component in many state-of-the-art
planar graph algorithms for shortest paths, minimum cuts, and maximum
flows. 
By combining these two techniques we remove the $\log n$ dependency
in the running time of the shortest-path algorithm, making it $O(n r^{-1/2} \log^2r)$.
\vspace{0.1in}

This work is part of a research agenda that aims to develop new
techniques that would lead to faster, possibly linear-time, algorithms
for problems such as minimum-cut, maximum-flow, and shortest paths with negative arc lengths.
As immediate applications, we show how to compute maximum flow in directed
weighted planar graphs in $O(n \log p)$ time, where $p$ is the minimum
number of edges on any path from the source to the sink. 
We also show how to compute any part of the DDG that corresponds to a
region with $r$ vertices and 
$k$ boundary vertices in $O(r \log k)$ time, 
which is faster than has been previously known for
small values of $k$. 
\end{abstract}

\newpage
\clearpage
\setcounter{page}{1}
\section{Introduction}


Finding shortest paths and maximum flows are among the
most fundamental optimization problems in graph theory.
On planar graphs, these problems are intimately related and can all be solved in linear or near-linear time.
Obtaining strictly-linear time algorithms for these problems is one of the main current goals of the planar graphs community~\cite{EisenstatK13,MinstCutSTOC}.
Some of these problems are known 
to be solvable in linear time (minimum spanning tree~\cite{MST}, shortest-paths with
non-negative lengths~\cite{HKRS97}, maximum flow with unit
capacity~\cite{EisenstatK13}, undirected unweighted
global min-cut~\cite{ChangL13}), but for many others, only nearly-linear time
algorithms are known. These include 
shortest-paths with negative
lengths~\cite{FR06,KMW10,ShayWulf}, 
multiple-source shortest-paths~\cite{Cabello,CabelloCE13,Klein05}, 
directed max-flow/min-$st$-cut~\cite{borradaile-klein-09,BKMNWN11,min-cut,Erickson2010,MinstCutSTOC}, and global
min-cut~\cite{ChFaNa04,LackiSankowski}. 

Many of the results mentioned above were achieved fairly recently, 
along with the development of more sophisticated shortest paths
techniques in planar graphs. In this paper we show how to combines 
two of these techniques: the technique of Henzinger, Klein, Rao, and
Subramanian for computing shortest paths with non-negative lengths in
linear time~\cite{HKRS97}, and the technique of Fakcharoenphol and Rao
to compute shortest paths on dense distance graphs in nearly
linear time in the number of vertices~\cite{FR06}.

The \emph{dense distance graph} (DDG) is an important element in many
of the algorithms mentioned above. It is a non-planar graph that
represents (exactly) distances among a subset 
of the vertices of the underlying planar graph $G$.
More precisely, an $r$-division~\cite{Frederickson} of an $n$-vertex
planar graph $G$, for some $r < n$, is a 
division of $G$ into $O(n/r)$ subgraphs (called \emph{regions}) $\{G_i\}_i$,  
where each region has at most $r$ vertices and $O(\sqrt{r})$ \emph{boundary} vertices,
which are vertices that the region shares with other regions.
There exist $r$-divisions 
of planar graphs with the additional property
that the boundary vertices in each region lie on a constant number of 
faces (called {\em holes})~\cite{FR06,ShayrDivision,Subramanian95}.
Consider an $r$-division of $G$ for some $r < n$.
Let $K_i$ be the complete graph on the $O(\sqrt{r})$ boundary
vertices of the region $G_i$, where the length of arc $uv$
is the $u$-to-$v$ distance in $G_i$. The graph $K_i$ is called the 
DDG of $G_i$. 
The union $\bigcup_i K_i$ is called the 
DDG of $G$ (or more precisely, the DDG of the given $r$-division of $G$).\footnote{DDGs are  similarly defined for decompositions which are not
$r$-divisions, but our algorithm does not necessarily apply in such cases.}

DDGs are useful for three main reasons. First, distances in the DDG of $G$ are the same as distances in $G$. Second,  it is possible to
compute shortest paths in DDGs in time that is nearly linear in the
number of vertices of the DDG~\cite{FR06} (FR-Dijkstra). I.e., in
sublinear time in the number of vertices of $G$. 
Finally, DDGs can be computed in nearly linear time either by invoking
FR-Dijkstra recursively, or  by using 
a multiple source shortest-path
algorithm~\cite{CabelloCE13,Klein05} (MSSP). 
Until the current work, the latter method was faster in all cases.

Since it was introduced in 2001, FR-Dijkstra has been used creatively as a
subroutine in many algorithms. These include algorithms for computing DDGs~\cite{FR06},
shortest paths with negative lengths~\cite{FR06}, 
maximum flows and minimum
cuts~\cite{BKMNWN11,min-cut,MinstCutSTOC,LackiSankowski}, 
and distance
oracles~\cite{Cabello,FR06,KaplanMNS12,Klein05,MozesSommer,Nussbaum11}. 
Improving FR-Dijkstra is therefore an important task with implications
to all these problems.
For example, consider the minimum $st$-cut problem in undirected
planar graphs. Italiano et. al~\cite{MinstCutSTOC}
gave an $O(n\log\log n)$ algorithm for the problem, improving the
$O(n\log n)$ of Reif~\cite{reif83} (using~\cite{HKRS97}). 
Three of the techniques used by Italiano et al. are: (i) constructing an
$r$-division with $r=\operatorname{polylog}(n)$  in $O(n\log r)$ time, 
(ii) FR-Dijkstra, 
and (iii) constructing the DDG in $O(n\log r)$. 
In a step towards a linear time algorithm for this fundamental problem,
Klein, Mozes and Sommer gave an $O(n)$ algorithm for constructing an
$r$--division~\cite{ShayrDivision}.  
This leaves the third technique as the only current bottleneck for
obtaining min-$st$-cut in $O(n)$ time.
The work described in the current
paper is motivated by the desire to obtain such a linear time
algorithm, and partially addresses techniques (ii) and (iii). 
It improves the running time of FR-Dijkstra, 
and, as a consequence, implies faster DDG construction,
although for a limited case which is not the one required
in~\cite{MinstCutSTOC}.

The linear time algorithm for shortest-paths with nonnegative lengths
in planar graphs of Henzinger, Klein, Rao, and
Subramanian~\cite{HKRS97} (HKRS) is another important result that has been
used in many subsequent algorithms. 
HKRS~\cite{HKRS97} differs from Dijkstra's
algorithm in that the order in which we {relax} the arcs is not determined by the
vertex with the current globally minimum label. Instead, it works with
a recursive division of the input graph into regions. It handles a
region for a limited time period, and then skips to another region. 
Within this time period it determines the order of
relaxations according to the vertex with minimum label in the {\em current}
region, thus avoiding the need to perform many operations on a large
heap. Planarity, or to be more precise, the existance of small
recursive separators, guarantees that local relaxations have limited
global effects. Therefore, even though some arcs are relaxed by HKRS
 more than once, the overall running time can be shown (by a
fairly complicated argument) to be linear.
Even though HKRS has been introduced
roughly 20 years ago and has been used in many
other algorithms, to the best of our
knowledge, and unlike other important planarity exploiting techniques,
it has always been used as a black box, and was not modified or
extended prior to the current work.\footnote{Tazari and M\"uller-Hannemann~\cite{TM09}
extended~\cite{HKRS97} to minor-closed graph classes, but that
extension uses the algorithm of~\cite{HKRS97} without change. It deals
with the issue of achieving constant degree in minor-closed classes of
graphs, which was overlooked in~\cite{HKRS97}.} 
 
\paragraph{\bf Our results.}
By combining the technique of Henzinger et al. with a modification of
 the internal building blocks of Fakcharoenphol and Rao's Dijkstra
 implementation, we obtain a faster
 algorithm for computing shortest paths on dense distance graphs.
This is the first asymptotic improvement
 over FR-Dijkstra. Specifically, for a DDG over an $r$--division of an
 $n$-vertex graph,
 FR-Dijkstra runs in $O(n\log(n) r^{-1/2} \log r)$. We remove the 
 logarithmic dependency on $n$, and present an algorithm whose running
 time is $O(n r^{-1/2} \log^2 r)$.
This improvement is useful in algorithms that use an $r$-division when $r$ is small (say $r=\operatorname{polylog}(n)$).
 
Our overall algorithm resembles that of HKRS~\cite{HKRS97}. However, in our
algorithm, the bottom level regions are not individual edges (as is
HKRS), but hyperedges, which are implemented by a suitably modified
version of {\em bipartite Monge heaps}. The bipartite Monge heap is 
 the main workhorse of Fakcharoenphol and Rao's
algorithm~\cite{FR06}. One of the main challenges in combining the two
techniques is that the efficency of Fakcharoenphol and Rao's algorithm
critically relies on the fact that the algorithm being implemented is Dijkstra's
algorithm, whereas HKRS does not 
implement Dijkstra's algorithm. To overcome this difficulty we modify
the implementation of the bipartite Monge heaps. We develop new Range
Minimum Query (RMQ) data structures, and a new way to use them to
implement bipartite Monge heaps.

Another  difficulty is that both the algorithm and the analysis of
HKRS had to be modified to work with hyperedges, and with the fact that
the relaxations implemented by the bipartite Monge heaps are performed
implicitly. 
On the one hand, using implicit relaxations causes limited
availability of explicit and accurate distance labels. On the other
hand, such explicit and accurate labels are necessay for the progress of the HKRS
algorithm as it shifts its limited attention span between different regions. 
We use an auxiliary construction and
careful coordination to resolve this conflict between fast implicit
relaxations and the need for explicit accurate labels.

\paragraph{\bf Applications.}
We believe that our fast shortest-path algorithm on the dense
distance graph is a step towards optimal algorithms for the important
and fundamental problems mentioned in the introduction. 
In addition, we describe two current applications of our improvement.
 In both applications, we obtain a speedup over previous algorithms by
 decomposing a region of $n$ vertices using an $r$-division and
 computing distances among the boundary vertices of the region in
 $O((n/\sqrt{r})\log^2 r)$ time using our fast shortest-path algorithm.  

\vspace{0.08in} \noindent  {\em Maximum flow when the source and sink are close.} In directed weighted planar graphs, 
we show how to compute maximum $st$-flow (and min-$st$-cut) in $O(n
\log p)$ time if there is some path from $s$ to $t$ with at most $p$
vertices. The parameter $p$ appears in the time bounds of several previous
maximum flow and minimum cut
algorithms~\cite{BH13,IS79,JV83,KaplanN11}.
Our $O(n \log p)$ time bound matches the fastest previously known
maximum flow algorithms in directed weighted planar graphs for $p =
\Theta(1)$~\cite{JV83} and for $\log p = \Theta(\log
n)$~\cite{borradaile-klein-09,Erickson2010}, and is asymptotically
faster than previous algorithms for all other values of $p$. See Section~\ref{section:maxflow}.

We believe that by combining our fast shortest-path algorithm with ideas
from the $O(n \log p)$ time min-$st$-cut algorithm of Kaplan and Nussbaum~\cite{KaplanN11} and
from the $O(n \log \log n)$ time min-$st$-cut algorithm of Italiano et al.~\cite{MinstCutSTOC}, we
can get an $O(n \log \log p)$ time min-$st$-cut algorithm for undirected planar graphs.
The details, which we were not able to compile in time for this submission,
will appear in a later version of this paper.

\vspace{0.08in} \noindent  {\em Fast construction of DDGs with few
boundary vertices.} The current bottleneck in various shortest paths
and maximum flow algorithms in planar graphs (e.g., min-$st$-cut in an undirected graph~\cite{MinstCutSTOC},
shortest paths with negative lengths~\cite{ShayWulf}) is computing all
boundary-to-boundary distances in a graph with $n$ vertices and $k\le
\!\sqrt{n}$ boundary vertices that lie on a single face.  
Currently, the fastest way to compute these $k^2$ distances is to use
the MSSP algorithm~\cite{CabelloCE13,Klein05}. After $O(n\log n)$
preprocessing, it can report the distance from any boundary vertex to
any other vertex  (boundary or not)  in $O(\log n)$ time, so all
boundary-to-boundary distances  can be found in $O((n + k^2) \log n) =
O(n\log n)$ time. We give an algorithm that computes the distances among
the $k$ boundary vertices in $O(n \log k)$ time. 
This algorithm can be used to construct a DDG of a region with $n$ vertices
and $k = O(\sqrt{n})$ boundary vertices in $O(n \log k)$ time.
In general, this does not improve the $O(n \log n)$ DDG construction time using MSSP
since typically $k=\Theta(\sqrt{n})$. However, there is an improvement when $k$ is
much smaller. 
For $k=\operatorname{polylog}(n)$, our algorithm takes $O(n\log\log
n)$ time. 

We conclude this section by discussing the interesting open problem of
computing a DDG of a region with $n$ vertices of which $k =
O(\sqrt{n})$ are boundary vertices. As we already mentioned, 
the conventional way of computing a
DDG is applying Klein's MSSP algorithm~\cite{Klein05,Erickson2010}, which requires
$O(n \log n)$ time regardless of the value of $k$. The MSSP algorithm
implicitly computes all the shortest-path trees rooted at vertices of
a face $\phi$ by going around the face $\phi$ and for every boundary
vertex $v$ of $\phi$ identifying the \emph{pivots} (arc changes) from
the tree rooted at the vertex preceding $v$ on $\phi$ to the tree
rooted at $v$. Eisenstat and Klein~\cite{EisenstatK13} showed an
$\Omega(n \log k)$ lower bound on the number of comparisons between
arc weights that any MSSP algorithm requires for identifying all the
pivots.\footnote{In fact they showed an $\Omega(n \log n)$ lower bound
by using a face $\phi$ with $\sqrt{n}$ boundary vertices; generalizing
the same proof for a face with $k$ vertices gives the $\Omega(n \log
k)$ bound.} Our algorithm can be used to compute the
DDG of a region in $O(n \log k)$ time without computing 
all the pivots. Note, however, that it does not break
the $\Omega(n \log k)$ lower bound. The DDG computation problem may be
an easier problem than the MSSP problem, since we are interested only
in the $O(k^2)$ distances among the boundary vertices and we are not
required to compute the pivots. Whether the DDG of a region can be
computed in $o(n \log k)$ remains an open problem. 

\paragraph{\bf Roadmap.} We begin in Section~\ref{ch:FR} with a
description of Fakcharoenphol and Rao's FR-Dijkstra algorithm and
the Monge heap data structure. This
description is essential since we modify the internal structure of
the Monge heaps to achieve our results.
In Section~\ref{warmup} we give a warmup improvement of FR-Dijkstra by
introducing a new RMQ data structure into the Monge heaps. 
This improvement decouples the logarithmic dependency on $n$ from the
logarithmic dependency on $r$.
In Section~\ref{sec:HKRS-FR} we eliminate the logarithmic dependency
on $n$ altogether by combining the shortest-path algorithm of
Henzinger et al. with modified Monge heaps similar to those 
described in the warmup. Finally, in Section~\ref{section:applications}
we give the details of two applications that use our algorithm.
 
For clarity of the presentation, we defer to an appendix some of the proofs which are less essential for the overall understanding of our algorithm.

\section{FR-Dijkstra~\cite{FR06}}\label{ch:FR}
In this section we overview FR-Dijkstra that emulates Dijkstra's
algorithm on the dense distance graph. Parts of our description
slightly deviates from the original description of~\cite{FR06} and
were adapted from~\cite{Book}. We emphasize that this description is
not new, and is provided as a basis for the modifications introduced
in subsequent sections.

Recall Dijkstra's algorithm. It maintains a heap of
vertices labeled by estimates $d(\cdot)$ of their distances from the
root. At each iteration, the vertex $v$ with minimum $d(\cdot)$ is
 {\sc Activated}: It is extracted from the heap, and all its adjacent arcs are relaxed. 
FR-Dijkstra implements {\sc ExtractMin} and {\sc Activate}
efficiently.

The vertices of  $K_i$ correspond to the $\sqrt r$ boundary vertices
of a region $R$ of an $r$--division with a constant number of holes. 
We assume in our description that $R$ has a single hole. There is a
standard simple way to handle a constant 
number of holes using the single hole case with only a constant factor
overhead (cf.~\cite[Section
5]{FR06}, \cite[Section 5.2]{KaplanMNS12}, and~\cite[Section 4.4]{ShayWulf}).
There is a natural cyclic order on the vertices of $K_i$ according to
their order on the single face (hole) of the region
$R$.

A matrix $M$ is a {\em Monge} matrix if for any pair of rows $i<j$ and columns $k<\ell$ we have that $M_{ik}+ M_{j\ell} \ge M_{i\ell }+ M_{jk}$. 
A {\em partial} Monge matrix is a Monge matrix where some of the
entries of $M$ are undefined, but the defined entries in each row
and in each column are
contiguous.
It is well known (cf.~\cite{KMW10}) that the upper and lower triangles
of the (weighted) incidence matrix $M$
of $K_i$ are {\em partial} Monge matrices.

To implement  {\sc ExtractMin} and {\sc Activate} efficiently, each
complete graph $K_i$ in the DDG is decomposed into 
complete bipartite graphs. 
The vertices of $K_i$
are first split into two consecutive halves $A$ and $B$,
 the complete bipartite graph on $A$ and $B$ is added to the
decomposition, and the same process is applied recursively on $A$ and
on $B$. Each vertex of $K_i$ therefore appears in $O(\log r)$
bipartite subgraphs. 
Furthermore,  and each bipartite subgraph of the decomposition
corresponds to a submatrix of $M$ that is {\em fully} Monge.

\begin{figure}[H]
\begin{center}
\includegraphics[scale=0.2]{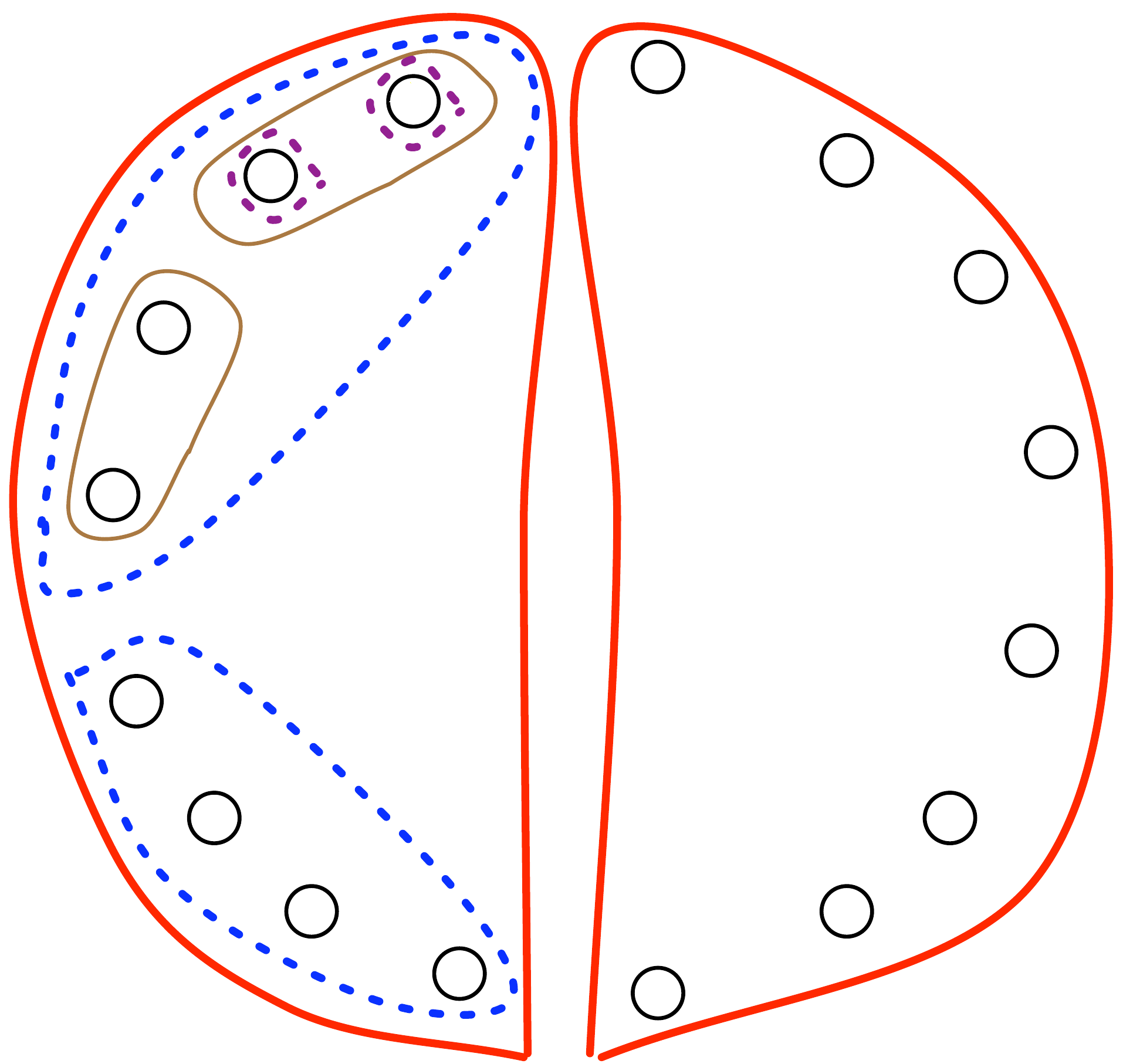}\ \ \ \ \ \ \ \ \ \ \ \ \ 
{\includegraphics[scale=0.8]{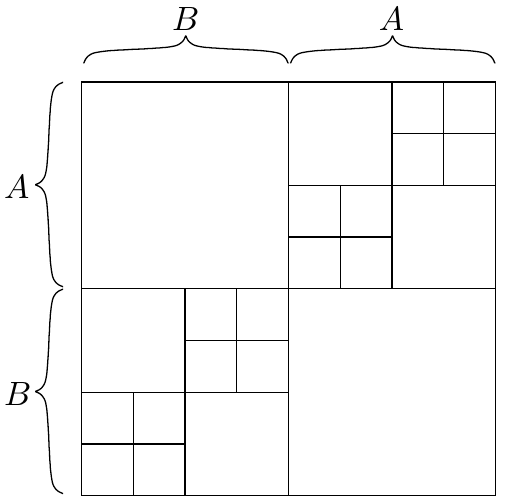}
}
\label{fig:bipartiteMonge}
\caption{The decomposition of $K_i$ into complete bipartite graphs
  (left) can also be viewed as a partition of the incidence matrix of
  $K_i$ into  Monge matrices (right).}
\end{center}
\end{figure}

\noindent

Let $\mathcal H = \{H_j\}$ denote the set of all
bipartite graphs in the decompositions of all $K_i$'s.
Note that $|{\mathcal H}| = O(n/\sqrt{r})$. 
The algorithm maintains a data-structure $\mathcal M_j$, called a {\em Monge Heap},
for each bipartite graph $H_j \in \mathcal H$.\footnote{In~\cite{FR06}
these are called bipartite Monge heaps. All the bipartite Monge heaps
that belong to the same complete graph $K_i$ are aggregated into a
single data structure which~\cite{FR06} call Monge heap. We do not use this aggregation.}  
Let $A, B$ be the bipartition of $H_j$'s vertices.
The Monge heap  supports:
 
\begin{list}{\labelitemi}{\leftmargin=1em \setlength{\itemsep}{1pt}}    
\item \proc{Activate}$(a,d)$
 - Sets the label of vertex $a \in A$ to be
  $d$, and implicitly relaxes all arcs incident to $a$. 
This operation may be called at most once per vertex.
\item \proc{FindMin} - Returns the vertex $v \in B$ with minimum
  label among the vertices not yet  extracted. 
\item \proc{ExtractMin} - Removes the vertex $v\!\! \in\!\! B$ with minimum
  label among the vertices not yet  extracted. 
\end{list}

The minimum element of every Monge heap
$\mathcal M_j$ is maintained in a regular global heap $\mathcal Q$.
In each iteration, a vertex $v$ with global minimum label is extracted from
$\mathcal Q$.
The vertex $v$ is then extracted from the Monge
heap $\mathcal M_{j}$ that contributed $v$, and the
new minimum vertex of $\mathcal M_{j}$ is added to $\mathcal Q$.
Note that since a vertex $v$ appears in multiple $H_j$'s,  $v$ may be
extracted as the minimum element of the global heap $\mathcal Q$
multiple times, once for each Monge heap it appears in\footnote{In the
  original description of FR~\cite{FR06}, vertices are never extracted
  from $\mathcal Q$. Instead, $\mathcal Q$ maintains only one
  representative for every region $R$ that is keyed by the minimum element
  of all (bipartite) Monge heaps of $R$.  When this representative
  is the minimum in $\mathcal Q$, they do not extract it form
  $\mathcal Q$ but instead increase its key to be the new (next) minimum of $R$. In other words, at any point in time, FR hold $n/r$ elements in $\mathcal Q$ (one element for each region) and overall it performs $(n/\sqrt{r}) \log r$ \proc{IncreaseKey} operations, while in our description $\mathcal Q$  holds $n/\sqrt{r}$  elements (one element for each Monge heap) and overall we perform $(n/\sqrt{r}) \log r$ \proc{ExtractMin} and \proc{Insert} operations.  The overall $O((n/\sqrt{r}) \log n \log r)$  running time is the same in both presentations.} . However,
the label $d(v)$ of $v$ is finalized at the first time $v$ is extracted. At that time, and only at that time, the algorithm marks $v$
as finalized, activates $v$ using
\proc{Activate} in {\em all} Monge heaps such that $v
\in A$, and updates the
representatives of those Monge heaps in $\mathcal Q$.

\paragraph{\bf Analysis.}
Since each vertex appears in $O(\log r)$ bipartite graphs $H_j$ in
$\mathcal H$, the number of times each vertex is extracted from the global heap $\mathcal
Q$ is $O(\log r)$. Since $\mathcal Q$ contains one representative
element from each Monge heap $\mathcal M_j$, a single call to {\sc ExtractMin} on 
$\mathcal Q$ takes $O(\log (n/\sqrt{r})) = O(\log n)$ time. Therefore,
the total time spent on extracting vertices from $\mathcal Q$ is $O((n/\sqrt{r})
\log n\log r)$. 

As for the cost of operations on the Monge heaps, {\sc Activate}
and {\sc ExtractMin} are called at most once per vertex in each
Monge heap, and the number of calls to {\sc FindMin} is bounded by
the number of calls to {\sc Activate} and {\sc
  ExtractMin}. We next show how to implement each of these operations in $O(\log r)$ time. 
 Since each vertex appears in $O(\log r)$ Monge
heaps, the total time spent
on operations on the Monge heaps is $O((n/\sqrt{r}) \log^2r)$.

\paragraph{\bf Implementing Monge heaps.}
Let $\mathcal{M}$ be the Monge heap of a bipartite subgraph with columns $A$ and $B$ and a corresponding
 incidence matrix $M$. 
For every vertex
$v\in B$ we maintain a bit indicating whether $v$ has already been
extracted from $\mathcal{M}$ (i.e., finalized) or not.

The distance label of a vertex $v \in B$ is defined to be $d(v)
= \min_{u \in A}\{d(u) + M_{uv}\}$  and is not stored explicitly.  
Instead, 
we say that $u \in A$ is the \emph{parent} of $v \in B$ 
if $d(u)$ is finite and  $d(v)
= d(u) + M_{uv}$.
As the algorithm progresses,
 the distance labels $d(u)$ of vertices $u\in A$ (hence the parents of vertices $v\in B$) may change.
We   maintain a binary search tree $T$ of triplets  $(a\! \in \! A,
b_1\! \in\!  B, b_2\! \in\!  B)$ indicating that $a$ is the current parent of all vertices of $B$  between $b_1$ and $b_2$. 
Note that a vertex $a\in A$ may appear
in more than one triplet because, after extracting  vertices of
$B$, the non-extracted vertices
of which $a$ is a parent might consist of several intervals. The Monge property of $M$ guarantees that if $a$ precedes $a'$ in $A$ then 
 all intervals of $a$ precede all intervals of $a'$.

Finally, the Monge heap structure also consists of a standard
heap $Q_B$ containing, for every triplet $(a, b_1, b_2) \in T$, a
vertex $b$ between $b_1$ and $b_2$ that minimizes
$d(b) = d(a) + M_{ab}$.\footnote{This heap can be implemented within the tree
$T$ by maintaining  subtree minima at the vertices of $T$.}

\begin{list}{\labelitemi}{\leftmargin=1em \setlength{\itemsep}{1pt}}    
\item{{\sc FindMin}}: Return a vertex $b$ with minimum distance label in $Q_B$.
    \item{{\sc ExtractMin}}: Extract the vertex $b$ with minimum distance label from $Q_B$ and mark  $b$
         as extracted. Find the (unique)
         triplet $(a, b_1, b_2)$ containing $b$ in $T$. Let $b'$ and $b''$ be the members of
         $B$ that precede and follow $b$, respectively,
within this triplet. The algorithm replaces the triplet
         $(a, b_1, b_2)$ with two triplets $(a,b_1, b')$ and $(a, b'', b_2)$ (if these
         intervals are defined). In each of these new triplets it
         finds the vertex $b^*$  that minimizes $d(b^*)=d(a)+M_{ab^*}$ and 
inserts $b^*$ into $Q_B$. The vertex $b^*$ is the minimum entry of $M$ in a given row and a range of columns and is found in $O(\log r)$ time using a naive one-dimensional static RMQ data structure on each column of $M$.
 
    \item{{\sc Activate}($a, d$)}: Set $d(a) = d$ and  find the
children of $a$ in $B$.
         If $a$ is the first vertex in the Monge heap for which
         this operation is applied then all the vertices of $B$ are children of $a$.
Otherwise, we show how to find the
children of $a$ whose current parent precedes $a$ in $A$ (the ones whose current parent follows $a$ in $A$ are found symmetrically).

 We traverse the triplets in $T$ one by one backwards starting with triplet $t=(w,f_1,f_2)$ such that $w$ precedes $a$ in $A$. We continue until we reach a triplet $t'=(u,b_1,b_2)$ such that $d(u) + M_{ub_2} < d(a) + M_{ab_2}$. If $t'=t$ we do nothing. If  we scanned all triplets preceding $t$ without finding $t'$ then the first child of $a$
 must belong to the last triplet $t''=(v,g_1,g_2)$ that we scanned. It is found  by a
binary search on the interval of $B$ between $g_1$ and $g_2$.
Otherwise, the
 first child of $a$
 must belong to the triplet $t''=(v,g_1,g_2)$ following $t'$, and can again be found via
binary search.

Let $x$ (resp., $y$) be the first (resp., last) child of $a$ in $B$,
as obtained in the preceding step. Note that there are no extracted vertices between $x$ and $y$ in $B$. This is because 
 we find the distances in
monotonically increasing order, and when we extract a vertex $b$ it is the minimum in the global heap $\mathcal
Q$  so it will never 
acquire a new parent. 
We therefore insert a new triplet $(a,x,y)$ into $T$. 

We remove from $T$ all other triplets
containing vertices between $x$ and $y$, and remove from $Q_B$ the
elements contributed by these triplets. Let $(u,b_1,b_2)$
be the removed triplet that contains $x$.
 If $x\not= b_1$ then we insert a new triplet $(u,b_1,z_1)$
 where $z_1$ is the vertex preceding $x$ in $B$.
 Similarly, if $(w,b'_1,b'_2)$  is the removed triplet that contains
 $y$ and $y \neq b_2'$ then we insert a new triplet 
 $(w,z_2,b_2')$ where $z_2$ is the vertex following $y$ in $B$.
  
 Finally, we update the three values that the new triplets $(a,x,y),(u,b_1,z_1),$ and $(w,z_2,b_2')$
 contribute to $Q_B$. We find these values by a range minimum query to the naive RMQ data structure.
\end{list}

\paragraph{\bf Analysis.}

Clearly, {\sc FindMin} takes $O(1)$ time. 
Both {\sc
ExtractMin} and {\sc Activate} insert a constant number of new
triplets to $T$ in $O(\log r)$ time, make a constant number of
range-minimum queries in $O(\log r)$ time, and update the
representatives of the new triplets in the heap $Q_B$ in $O(\log r)$ time.
{\sc Activate}($a,d$), however, may
traverse many triplets to identify the children of $a$.
Since 
all except at most two of the triplets that it traverses are removed, 
we can charge their traversal and removal
to their insertion in a previous {\sc ExtractMin} or {\sc
Activate}.

\section{A Warmup Improvement of FR-Dijkstra}\label{warmup}

In this section we show how to modify FR-Dijkstra from
$O((n/\sqrt{r})(\log n \log r))$ to $O((n/\sqrt{r})(\log n + \log^2 r \log\log r))$. This improves FR-Dijkstra for small values of $r$ and is obtained  by avoiding the vertex copies in the global heap $\mathcal Q$ using a decremental RMQ data structure.

Recall that, at any given time, the global heap $\mathcal Q$ of FR-Dijkstra maintains $O(n/\sqrt{r})$ items -- one item for each Monge heap. 
However, each vertex has copies in $O(\log r)$ Monge heaps so overall
$O((n/\sqrt{r}) \log r)$ items are extracted from  $\mathcal Q$. Each
extraction takes $O(\log (n/\sqrt{r})) = O(\log n)$ time so the
complexity of FR-Dijkstra is $O((n/\sqrt{r}) \log r \log n)$. 
In our modified algorithm the global heap $Q$ contains, for each
vertex $v$, a single item whose key is the minimum label over all copies of $v$.
In other words,  $\mathcal Q$  maintains a total of $O(n/\sqrt{r})$ items throughout
the entire execution. An
item that corresponds to a vertex $v$ is extracted from $Q$ only once
in $O(\log n)$ time, but when it is extracted it may incur $O(\log r)$
calls to  {\sc DecreaseKey}. We use a Fibonacci heap~\cite{FT87} for
$\mathcal Q$ that only takes constant time for {\sc DecreaseKey}. Note
that the total number of operations on $Q$ is still $O((n/\sqrt{r})\log r)$.

The main problem with having one copy is that 
a triplet $(a,x,y)$ might now contain extracted vertices between $x$ and $y$ in $B$. 
The original implementation of FR-Dijkstra uses an elementary RMQ data structure (a binary search tree for each row of the $\sqrt{r}\times \sqrt{r}$ matrix $M$). This data structure is  only queried on intervals $[x,y]$ that have no extracted vertices. For such queries one may also  use the
RMQ data structure of Kaplan et
al.~\cite{KaplanMNS12}, which has the same $O(\log r)$ query time and requires only  $O(\sqrt{r} \log r)$
construction time and space.
To cope with query intervals $[x,y]$ that have extracted vertices we next present a {\em dynamic} RMQ data structure that can handle extractions. 

\begin{lemma}\label{lemma:RMQ}
Given an $n\times n$ partial Monge matrix $M$, one can construct in
$O(n\log n)$ time a dynamic data structure that supports the
following operations: (1)  in  $O(\log^3 n)$ time, set all entries of
a  given column as inactive  (2)  in $O(\log^3 n)$ time, set all entries
of a given column as active (3)  in  $O(\log^2 n)$ time, report the
minimum active entry in a query row and a contiguous range of columns
(or $\infty$ if no such entry exists). \\
For a decremental data structure in which only operations 1 and 3 are
allowed, operation 1 can be supported in amortized $O(\log^2 n \log\log n)$  time and operation 3 in worst case $O(\log n \log\log n)$  time.
\end{lemma}

\noindent The proof of Lemma~\ref{lemma:RMQ} appears in
Appendix~\ref{ap:RMQproof}. Note that we stated the above lemma for partial Monge
matrices and not full Monge matrices. This is because we apply the RMQ
data structure to the upper and lower triangles of the  $\sqrt{r} \times\! \sqrt{r}$ incidence
matrix $M$ of $K_i$ (which are partial Monge) as opposed to FR-Dijkstra that
uses a separate RMQ for each bipartite subgraph (whose corresponding
submatrix is full Monge) of $K_i$'s decomposition. 
In this section we only deactivate columns of $M$ and never activate columns. Therefore, we use the second data
structure in the lemma. The first data
structure in Lemma~\ref{lemma:RMQ} is used in Section~\ref{sec:HKRS-FR}.

\paragraph{\bf Modifying FR-Dijkstra.} We already mentioned that we
want  $\mathcal Q$ to be a Fibonacci heap and to include one item
per boundary vertex $v$. We also change the Monge heaps so that $Q_B$ now contains, for every triplet $(a, b_1, b_2)$, a vertex that minimizes
$d(b)=d(a)+M_{ab}$  among vertices $b$ between $b_1$ and $b_2$ that
where {\em not yet extracted from $\mathcal Q$}. If all vertices
between $b_1$ and $b_2$ were already extracted then the triplet has no
representative in $Q_B$. 

As before, in each iteration, a minimum item, corresponding to a vertex $v$ is extracted from
$\mathcal Q$ in $O(\log n)$ time. In FR-Dijkstra, $v$ is then
extracted from the (unique) Monge heap $\mathcal M_{j}$ that
contributed $v$ and  the new minimum vertex of $\mathcal M_{j}$  is
added to $\mathcal Q$. In our algorithm, $v$'s extraction from $\mathcal Q$
affects all  the $O(\log r)$  $\mathcal M_{j}$'s that include $v$ in
their $B$. Before handling any of them we  
apply operation (1) of the RMQ data structure of Lemma~\ref{lemma:RMQ}
in $O(\log^2 r \log\log r)$ time. Then, for each such  $\mathcal
M_{j}$, if $v$ was the minimum of some triplet $(a, b_1, b_2)$ then we
apply the following operations: (I)  remove $v$ from $Q_B$ in $O(\log
r)$ time, (II) query in $O(\log r  \log\log r)$ time the RMQ data
structure of Lemma~\ref{lemma:RMQ} for vertex $b^* = RMQ(b_1,b_2)$,
(III) insert $b^*$ into $Q_B$  in $O(\log r)$ time. Note that we do
not replace the triplet $(a, b_1, b_2)$ with two triplets as FR do in
{\sc ExtractMin}.  
Finally, if $v$ was also the minimum of $Q_B$ then let $w$ be the new
minimum of $Q_B$. If $w$'s value in $\mathcal Q$ is larger than in
$Q_B$ then we update it in $\mathcal Q$ using {\sc DecreaseKey} in
constant time.  

Next, we need to activate $v$ in all $\mathcal M_{j}$'s  that have
$v$ in their $A$. This is done exactly as in FR-Dijkstra by removing (possibly many) triplets from $T$ and inserting three new triplets.  
We now find the minimum in each of these three triplets by querying the RMQ data structure of Lemma~\ref{lemma:RMQ}. 
Finally, we update $Q_B$ and if this changes the minimum element $w$ of $Q_B$ and $w$'s value in $\mathcal Q$ is larger than in $Q_B$ then we update it in $\mathcal Q$ using {\sc DecreaseKey} in constant time. 

To summarize,  each vertex $v$ of the $O(n/\sqrt{r})$ vertices is
extracted once from $\mathcal Q$ in $O(\log n)$ time. For each such
extraction we do a single RMQ column deactivation in $O(\log^2 r \log\log
r)$ time, we do $O(\log r)$ RMQ queries in total $O(\log^2 r \log\log
r)$  time incurring $O(\log r)$
{\sc DecreaseKey} operations on $\mathcal Q$ in total $O(\log r)$
time, we activate $v$ in  $O(\log^2 r)$ time, and finally 
we do $O(\log r)$
updates to $Q_B$  in total $O(\log^2 r)$ time incurring $O(\log r)$
{\sc DecreaseKey} operations on $\mathcal Q$ in total $O(\log r)$ time . The overall complexity is thus $O((n/\sqrt{r})(\log n + \log^2 r
\log\log r))$. 


\section{The Algorithm}\label{sec:HKRS-FR}
In this section we describe our main result: Combining the modified
FR-Dijkstra from the previous section with HKRS~\cite{HKRS97} to
obtain a shortest-path algorithm that runs in $O((n/\sqrt{r}) \log^2 r)$ time. 
 
\subsection{Preprocessing}
Our algorithm first performs the following preprocessing steps (described below in more detail). 

\begin{algorithm} \caption{{\sc Preprocess}$(G,r)$}
\label{alg:preprocess}
\begin{algorithmic}[1]
\REQUIRE $G$ is an $n$-vertex directed planar graph with non-negative arc lengths
and degree at most 3. $r<n$.
\STATE let $\vec{r} = (r_1,r_2, \dots, r_k)$ be such that $r_1 = r$  and
$r_i = r_{i-1}^2$ for every  $i>1$.
\STATE compute a recursive $\vec{r}$--division of $G$ with a constant number of
holes
\FOR {each region $R$ of the $r_1$--division}
 \STATE compute the dense distance graph of $R$
 \STATE initialize all bipartite Monge heaps of $R$
\ENDFOR
\end{algorithmic}
\end{algorithm}

A recursive $\vec{r}$--division with $\vec{r} = (r_1,r_2, \dots, r_k)$ is a decomposition tree of $G$ in which the root corresponds to the entire graph $G$ and the nodes at height $i$ correspond
to regions that form an $r_i$--division of $G$. Throughout the paper
we use $r$-divisions with a constant number of holes. 
Our algorithm differs from HKRS~\cite{HKRS97} in our choice of $r_i =
r_{i-1}^2$ (in HKRS $r_i = 2^{O(r_{i-1})}$). This results in a
decomposition tree of height roughly $\log\log n$ (in HKRS it is
roughly $\log^* n$).\footnote{All logarithms in this paper are base 2.}
For $i \geq 1$, the regions of the $r_i$--division
are called height-$i$ regions. The height of a vertex $v$ is defined
as the largest integer $i$ such that $v$ is 
 a boundary vertex of a height-$i$ region. 
Our algorithm computes the recursive $\vec{r}$--division in linear
time using the algorithm of~\cite{ShayrDivision}. 
We assume the input graph $G$ has degree at most three.\footnote{This
  is a standard assumption that can be obtained in linear time by
  triangulating the dual graph with infinite length edges.} With this
assumption the algorithm of~\cite{ShayrDivision} can be easily modified to return an
$\vec{r}$--division in which, for every
$i$, each vertex $v$ belongs to at most 3 height-$i$ regions. 

We next describe the height-0 regions. 
In~\cite{HKRS97}, height-0 regions comprise of individual edges. 
We define height-0 regions  differently.
Consider a height-1 region $R$ (that is, $R$ is a region with $r_1=r$
edges and $O(\sqrt r)$ boundary vertices). The DDG of $R$ is a complete
weighted graph over the boundary vertices of $R$, so it has $O(\sqrt{r})$ vertices and 
$O(r)$ edges.
Consider the set of Monge Heaps (MHs) of $R$.  
Each MH corresponds to a bipartite graph with left side $A$ and right side $B$. Each boundary vertex $v$ of $R$ appears in $O(\log r)$ MHs. 
For each occurrence of $v$ on the $B$ side of some MH $h$, we create 
a copy $v_h$ of $v$.
Vertices of $G$ such as $v$ are called {\em natural vertices}, and we
say that $v_h$ is called a {\em copy} of $v$ (each natural vertex has
$O(\log r)$ copies).

We add a zero length arc from every $v_h$ to $v$. 
Each of these newly added arcs is the single element in a distinct
height-0 region  $R$. The label of such an arc (and hence the label
of the corresponding region $R$) is defined (as
in~\cite{HKRS97}) to be the label of its tail, and is initialized to infinity. 
 
Denote the bipartition of the vertex set of an MH $h$ by $A_h,B_h$. 
Let $h$ be a MH.
For every vertex $v\in A_h$, we construct a hyperarc $e$ (directed
hyperedge)  whose tail is $v$ and whose heads are $\{w_h : w \in
B_h\}$. We associate $e$ with the MH $h$.
Each hyperarc $e$ is the single element in a distinct height-0 region of
$R$. The label of this height-0 region is defined to be the label of
the tail of $e$, and it is initialized to infinity. 
Note that our construction is such that the arcs of the DDG do not
belong to any region $R$ in the algorithm. Rather, an arc of the  DDG
is represented in the algorithm in two ways; By a hyperarc, and by the Monge heap to
which it belongs. Therefore, from now on we may refer to and think of an arc $vw$ of
the DDG as an arc $vw_h$, and to a vertex $w \in B_h$ as the copy $w_h$.

\begin{observation}\label{obs:height0}
The height of every copy vertex $v_h$ is 0.
\end{observation}
\begin{proof}
The construction is such that all arcs and hyperarcs
incident to a vertex $v_h$ correspond to height-0 regions of the same
height-1 region
$R$. Therefore, $v_h$ is not a boundary vertex of $R$, so the height of
$v_h$ is 0. \qed
\end{proof}

Let $c_1, c_2$ be the constants such that an $r$-division
of an $n$-vertex graph has at most $c_1 n/r$ regions, each having at most
$c_2 \sqrt{r}$ boundary vertices.
Let $c_5$ be the constant such that there are at most $c_5 \log r_1$
height-0 regions involving $v$.

\begin{lemma} \label{lem:num-0-regions}
The number of height-0 regions is $c_1c_2c_5\frac{n}{\sqrt r_1}\log r_1$
\end{lemma}
\begin{proof}
There are at
most $c_1\frac{n}{r_1}$ regions $R$ of height-1. A height-1 region $R$ 
 has at most $c_2\sqrt{r_1}$ boundary vertices. 
Each boundary
 vertex $v$ of $R$ appears in $O(\log r_1)$ MHs, so it has $O(\log
 r_1)$ copies. 
Hence there are $O(\log r_1)$ height-0 regions that
 correspond to individual arcs $v_hv$, and $O(\log r_1)$ height-0
 regions that correspond to individual hyperarcs whose tail is $v$.
There are at most $c_5 \log r_1$
height-0 regions involving $v$. The total number of height-0 regions
is therefore $c_1c_2c_5\frac{n}{\sqrt r_1}\log r_1$. \qed
\end{proof}

\subsection{Specification of Monge Heaps} \label{section:MH-spec}

Let $h$ be a MH. The arcs of the DDG represented by $h$ all have their tails in $A_h$ and heads in
$B_h$. 
Let $d(v)$ denote the current label of vertex $v$, and let
$\ell(vw)$ denote the length of the $v$-to-$w$ arc in the DDG.
Each vertex $v \in A_h$ has an associated set of consecutive vertices
in $B_h$ such that $\forall v' \in A_h , w_h \in B_h$ we have $d(v) + \ell(vw)
\leq d(v') + \ell(v'w)$. 
The vertices of $B_h$ in the interval
associated with a vertex $v$ of $A_h$ are called {\em the children of
  $v$}, and $v$ is called {\em the parent} of these vertices.
At any given time a vertex $w_h \in B_h$ is the child of exactly one
vertex $v \in A_h$ (the only exception is that upon initialization, no
vertex is assigned a parent). 
A vertex $w_h \in B_h$ can be either {\em active} or {\em inactive}. 
Initially all vertices are active.
A MH supports the following operations:
\begin{itemize}
  \item{{\sc FR-Relax}$(h, v)$} - Implicitly relax all arcs of MH $h$
  emanating from vertex $v \in A_h$. Internally, adds to the set
  of children of $v$ all the vertices in
  $B_h$  whose labels decreased because of these relaxations. This
  operation activates any of the newly acquired
  children of $v$ that were inactive.
\item{{\sc FR-GetMinChild}$(h, v)$} - Returns the active child $w_h$
  of $v$ in MH $h$ minimizing $\ell(vw_h)$.
\item{{\sc FR-Extract}$(v_h)$} - Makes vertex $v_h$ inactive  in MH $h$. Let
  $u$ be the parent of $v_h$ in MH $h$.  Returns the active child $w_h$
  of $u$ in $h$ minimizing $\ell(uw_h)$.
\end{itemize}

We discuss two implementations of these Monge heaps in
Appendix~\ref{ap:MHimplement}. These implementations are summarized in
the following lemma. 

\begin{lemma}\label{lem:HKRS-RMQ}
There exist two implementations of the above MH
with the following construction time (for all MHs of all regions in an $r$--division of an $n$-vertex graph), and
operation times:
\begin{enumerate}
\item  
$O((n/\sqrt{r}) \log^2 r)$ construction time and 
$O(\log^3 r)$ time per operation.
\item 
$O(n\log r )$ construction time and $O(\log r)$ time per operation.
\end{enumerate}
\end{lemma}

In some applications, in particular in those that work
with \emph{reduced lengths} (such as the maximum flow algorithm of 
section~\ref{section:maxflow}), shortest path computations on the DDG
are performed multiple
times with slightly different DDGs. In each time, the new DDG can be
obtained quickly but the RMQ data structure needs to be built from
scratch. For such applications, the fast construction time of the
first implementation is crucial.
For other applications which perform multiple shortest path
computations on the same DDG such as the application in
Section~\ref{sec:DDGapp}, the second implementation is appropriate, and
yields slightly faster shortest paths computations. The $O(n\log r)$-time
construction  is not a bottleneck of such applications since
it is performed once, and matches the currently fastest known construction of the DDG.
In what follows we denote the time to perform an MH operation as
$O(\log^{c_q} r)$ where the constant $c_{q}$ is either 3 or 1
depending on the implementation choice. 

\subsection{Computing Shortest Paths on the DDG}

The pseudocode of the algorithm is given below. Parts of the
pseudocode  appear in black font and parts in blue font. 
The black parts describe our algorithm, and at this point in the paper it is enough to focus only on them. The blue parts are additions
to the algorithm that should not (and in fact cannot) be implemented. They will be used later (section~\ref{sec:analysis}) 
for the analysis.

The algorithm maintains a label $d(v)$ for each vertex $v$. Initially
all labels are infinite, except the label of the source $s$, which is 0.
For each height-$i$ region $R$, the algorithm maintains
a heap $Q(R)$ (implemented using Fibonacci heaps~\cite{FT87})
containing the height-($i-1$) subregions of $R$.
For a height-0 region $R$, $Q(R)$ contains
the single element (arc or hyperarc)  $e$ comprising the region $R$.  
If $v$ is the tail of $e$, then the key of this single element is $d(v)$ if $e$ is not
relaxed, and infinity otherwise. 
For $i\ge 1$ a height-$i$ region $R$, the key of a 
 subregion $R'$ of $R$ in $Q(R)$ is $\minKey(Q(R'))$.
Such heaps were also used in~\cite{HKRS97}
but they were not implemented by Fibonacci heaps. 
The heaps support the following operations: 
$\getKey, \minKey, \minItem,
\decreaseKey, \increaseKey$. The first four operations take constant
amortized time, and the fourth takes logarithmic amortized time.

\newcommand{\red}[1]{\textcolor{blue}{#1}}
\renewcommand{\algorithmicreturn}{\red{\textbf{return}}}

\begin{algorithm} \caption{{\sc Process}$(R\red{,\debt})$}
\label{alg:process}
\begin{algorithmic}[1]
\IF {$R$ is a height-0 region that contains a single hyperarc $e$}
   \STATE let $v$ be the tail of $e$
   \STATE let $h$ be the Monge heap to which $e$ belongs
   \STATE {\sc FR-Relax}$(h,v)$  \label{line:FRrelax}
   \STATE $w_h \leftarrow $ {\sc FR-GetMinChild}$(h,v)$ \label{line:minchild}
   \STATE \red{$\debt \leftarrow \debt + \log^{c_q} r_1$} \label{line:hyperdebt}
   \STATE $d(w_h) \leftarrow d(v) + $ length of arc $vw$ in the DDG of $h$  \label{line:minchild-relax}
   \STATE let $R'$ be the height-0 region consisting of the arc $w_h w$
   \STATE \red{$\debt \leftarrow \debt\ + $ }{\sc Update}$(R',w_hw,d(w_h)\red{,w_h})$  \label{line:minchild-update}
   \STATE $\increaseKey(Q(R), e, \infty)$
\ELSIF{$R$ is a height-0 region that contains a single arc $v_h v$}
   \IF {$d(v) > d(v_h)$}
       \STATE $d(v) \leftarrow d(v_h)$ \label{line:v-relax}
       \FOR {each hyperarc $e$ whose tail is $v$} 
          \STATE \red{$\debt \leftarrow \debt\ +$ }{\sc Update}$(R(e),  e, d(v)\red{,v})$ \label{line:update}
       \ENDFOR
   \ENDIF 
   \STATE $w_h \leftarrow$ {\sc FR-Extract}$(v_h)$ \label{line:FRextract}
   \STATE \red{$\debt \leftarrow \debt  + \log^{c_q} r_1$}   \label{line:arcdebt} 
   \STATE $d(w_h) \leftarrow d(v) + $ length of arc $vw$ in the DDG of $h$  \label{line:minchild_a}
   \STATE let $R'$ be the height-0 region consisting of the arc $w_h w$
   \STATE \red{$\debt \leftarrow \debt\ +$ }{\sc Update}$(R',w_hw,d(w_h)\red{,w_h})$ \label{line:minchild_b} 
   \STATE $\increaseKey(Q(R), v_h v, \infty)$
\ELSE  
   \STATE \red{$\upDebt \leftarrow 0\ ;\ \credit \leftarrow 0$}
   \FOR {$\alpha_{\heit(R)}$ times or until $\minKey(Q(R))$ is infinity}
      \STATE $R' \leftarrow \minItem(Q(R))$
      \STATE \red{$\credit \leftarrow \credit + \debt/\alpha_{\heit(R)}$}
      \STATE \red{$\upDebt \leftarrow \upDebt\ +\ $}{\sc
        Process}$(R'\red{,\debt/\alpha_{\heit(R)} + \log|Q(R)|})$ \label{line:proc}
      \STATE $\increaseKey(Q(R),R',minKey(Q(R')))$\label{line:updateR'}
   \ENDFOR
   \STATE \red{$\debt \leftarrow \debt + \upDebt - \credit$}
   \STATE \red{\textbf{if}} {\red{$\minKey(Q(R))$ will decrease in the
       future}} \red{\textbf{then}} \label{line:stable}
   \STATE \ \ \ \red{\textbf{return }}\red{$\debt$}
   \STATE \red{\textbf{else}} \red{(this invocation is stable)}
    \STATE \ \ \ \red{pay off debt by withdrawing from account of
        $(R,\entry(R))$}
    \STATE \ \ \ \red{\textbf{return} 0 }
   \STATE \red{\textbf{endif}}
\ENDIF
\end{algorithmic}
\end{algorithm}

\begin{algorithm} \caption{{\sc Update}$(R,x,k\red{,v})$}
\label{alg:update}
\begin{algorithmic}[1]
\REQUIRE $R$ is a region, $x$ is an item of $Q(R)$, $k$ is a key
value\red{, $v$ is a boundary vertex of $R$}.
\STATE $\decreaseKey(Q(R),x,k\red{,v})$
\IF {the $\decreaseKey$ operation reduced the value of $\minKey(Q(R))$}
   \STATE \red{$\entry(R) \leftarrow v$} \label{line:entry}
   \RETURN $\red{1 + }\text{\sc Update}(\parent(R),R,k\red{,v})$ \label{line:upd1}
\ELSE \RETURN \red{1} \label{line:upd2}
\ENDIF
\end{algorithmic}
\end{algorithm}
\renewcommand{\algorithmicreturn}{\textbf{return}}
\renewcommand{\algorithmicelse}{\textbf{else}}

As in~\cite{HKRS97}, the main procedure of the algorithm is the procedure {\sc
  Process}, which processes a region $R$. The algorithm repeatedly
calls {\sc Process} on the region $R_G$ corresponding to the entire graph,
until $\minKey(Q(R_G))$ is infinite, at which point the labels
$d(\cdot)$ are the correct shortest path distances. 
As in~\cite{HKRS97}, processing a height-$i$ region $R$ for $i \geq 1$
consist of calling {\sc Process} on at most
$\alpha_i$ subregions $R'$ of $R$ (Line~\ref{line:proc}). Processing a region $R'$ may change
$\minKey(Q(R'))$, so the algorithm updates the key of $R'$ in $Q(R)$
in Line~\ref{line:updateR'}.\footnote{Since edge lengths are non-negative, and since
the labels $d(\cdot)$ (and hence keys) only change by relaxations or
by setting to infinity, $\minKey(Q(R'))$ may only increase when
processing $R'$. Therefore updating the key of $R'$ in $Q(R)$ is done
using $\increaseKey$.}

Our algorithm differs from~\cite{HKRS97} in processing height-0
regions. Let us first recall how~\cite{HKRS97} processes height-0
regions.
Each height-0 region $R$ corresponds to a
single arc $uv$, and processing $R$ consists of relaxing $uv$. 
If the relaxation decreases $d(v)$, then all arcs whose tail is $v$
become unrelaxed. For every such arc $e'$ the key of the corresponding
height-0 region $R'$ corresponding to $e'$ needs to be updated to $d(v)$.
Doing so may require updating the keys of ancestor regions of
$R'$. The procedure {\sc Update} performs this chain of updates.

In our case, the height-0 regions may correspond to single arcs or to single hyperarcs. If a height-0 region $R$ consists of a single hyperarc
$e$ whose tail is $v$, then processing $R$ implicitly relaxes all the
arcs of the DDG represented by $e$ by calling {\sc FR-Relax}
(Line~\ref{line:FRrelax}).  We can
only afford to update explicitly the
label of a single child of $v$. 
We do so for the child  $w_h$ of $v$ with minimum label 
(Line~\ref{line:minchild-relax}). 
Since $w_h$ is a copy vertex of a natural vertex $w$, there is exactly one arc whose tail is
$w_h$ (namely, $w_hw$), and no hyperarcs whose tail is $w_h$. The arc
$w_hw$ is now unrelaxed, so the key of the height-0 region
corresponding to $w_hw$ needs to be updated by a call to {\sc Update}
(Line~\ref{line:minchild-update}).

If a height-0 region $R$ consists of a single arc $v_hv$, then, as
in~\cite{HKRS97}, processing $R$ explicitly relaxes all elements
(hyperarcs in our case), whose tail is $v$, and calls {\sc Update} to
update the keys in the heaps of the corresponding height-0 regions and their
ancestor regions (Lines~\ref{line:v-relax}--~\ref{line:update}). 
Similarly to the implementation of FR-Dijkstra, since at this point
the vertex $v_h$ has no outgoing unrelaxed arcs, the algorithm 
extracts $v_h$ from the Monge heap $h$ to which it belongs (so $v_h$ becomes
inactive). Among the children of $v$ there is now a new minimum active
child $w_h$ ($v_h$ may have been the previous minimum, but it has just
become inactive). The algorithm handles $w_h$, as in the previous case, by updating the key of the height-0 region
corresponding to $w_hw$ with a call to {\sc Update}
(Line~\ref{line:minchild_b}).

\subsection{Correctness}
Distances in the input graph and in the graph constructed by our
algorithm are identical since duplicating vertices and adding
zero-length arcs between the copies does not affect distances.

We note that the correctness of the algorithm in~\cite{HKRS97} does not
depend on the choice of attention span (the parameters
$\alpha_i$). These parameters only affect the running time. In
particular, an implementation of that algorithm in which the attention
span is not fixed would yield the correct distances upon termination.

We argue the correctness of our algorithm by considering such a variant
of the algorithm of~\cite{HKRS97} on the graph constructed by our
algorithm, where each hyperarc is represented by the individual arcs
forming it. We refer to this variant as algorithm H. 
In what follows we still refer to hyperarcs to make
the grouping of the individual arcs clear. We say that an arc $vw_h$
belongs to hyperarc $e$ if $vw_h$ is one of the arcs forming $e$.

Let $e$ be a hyperarc whose tail is $v$. Let $R$ be the height-1 region containing
$e$. Consider a time in the execution of algorithm H when $R$ is processed and
chooses to process a height-0 region corresponding to a single
arc $vw_h$ that belongs to $e$. This implies that the label of $vw_h$
is the minimum among all height-0 regions (edges) in $R$.
Observe that at this time all other arcs of $e$ have the
same label as that of $vw_h$ (they might have infinite label if they have already been
processed). Furthermore, since lengths are non-negative, processing
$vw_h$ cannot decrease the label of $w_h$ below that of $v$. 
We define algorithm H to process all arcs corresponding to $e$ one
after the other, regardless of the attention span. The above
discussion implies that algorithm H terminates with the correct
distance labels for all vertices in the graph.

We now show that algorithm H and our algorithm relax exactly the same
arcs, and in the same order. 
The only difference between algorithm H and our algorithm is that our
algorithm performs implicit
relaxations using Monge heaps, while algorithm H performs all
relaxations explicitly. In particular, when relaxing a hyperarc $e$
whose tail is $v$, our algorithm only updates the label of the minimum
child of $v$ (Lines~\ref{line:minchild}--\ref{line:minchild-relax}),
whereas algorithm H updates the labels of all children of $v$.
We call a vertex $w_h$ {\em latent} if its label does not reflect prior relaxations of arcs
$vw_h$ incident to it (because those relaxations were implicit). 
A vertex that is not latent is {\em accurate}.
Note that whenever the label of a vertex $w_h$ is updated by our
algorithm (i.e., $w_h$ becomes accurate), the label of the height-0
region consisting of the unique arc 
$w_hw$ is also updated
(Lines~\ref{line:minchild}--\ref{line:minchild-update}).
 
To establish that the order of relaxations is the same in both
algorithms, it suffices to establish the following lemma:
\begin{lemma}
The following invariants hold:
\begin{enumerate}
\item Every latent vertex is active.
\item Let $h$ be an MH, and let $v$ be a vertex in $A_h$. 
Let $S_v$ be the set of active children of $v$ with minimum
label. If $S_v$ is non-empty then there is a vertex in $S_v$ that is
accurate. 
\item If the key of a height-0 region consisting of a single arc $w_hw$
  is finite then $w_h$ is active.
\end{enumerate}
\end{lemma}

\begin{proof}
The proof is by induction on the calls to {\sc Process} on height-0
regions performed by our algorithm.
Initially there are no latent vertices, all the vertices on the $B$
side of any MH are active, and all height-0 regions consisting of a
single arc have infinite labels, so all invariants trivially hold.
The inductive step consists of two cases. 

Suppose first that the height-0 region being processed consists of a
single hyperarc $e$. This hyperarc is implicitly relaxed in
Line~\ref{line:FRrelax}. 
All the vertices whose label should be changed by the relaxation are
guaranteed to become active (by the specification of {\sc FR-Relax}). Among
these vertices, $w_h$, the vertex with minimum label, is explicitly relaxed
(Lines~\ref{line:minchild}--\ref{line:minchild-relax}), so it is
accurate. All the others are latent. This establishes that the first
two invariants are maintained. 
In Line~\ref{line:minchild-update} the key of the height-0 region $R'$
consisting of the unique arc whose tail is $w_h$ is set to a finite value. Since
$w_h$ has just become active, the third invariant is also maintained.

Suppose now that the height-0 region being processed consists of a
single arc $v_hv$. By the third invariant, $v_h$ is active. By the
second invariant, $v_h$ must be accurate.
The arc $v_hv$ is explicitly relaxed, which does not create any
new latent vertices. In Line~\ref{line:FRextract}, $v_h$ is
deactivated. This is the only case where a vertex is
deactivated. Since $v_h$ is accurate, invariant~1 is maintained. 
Let $u$ be the parent of $v_h$ in MH $h$. 
By the specification of {\sc FR-Extract}, it returns an
active child $w_h$ of $u$ whose label is minimum among all of $u$'s active
children. The label of $w_h$ is explicitly updated in
Line~\ref{line:minchild_a}, so invariant~2 is maintained.
The height-0 region whose label is set to a finite value in Line~\ref{line:minchild_b}
consists of the single arc $w_hw$. Since $w_h$ is active, invariant~3
is maintained. \qed
\end{proof}

Note that invariants~1 and~2 guarantee that in every Monge heap $h$, if there exists a latent vertex $w_h$
then there exists an accurate vertex $u_h$ whose label is at most the correct label of $w_h$ (by correct label we mean the label that $w_h$ would have gotten had we performed the relaxations explicitly). This implies that as long
as a vertex is latent, its label does not affect the running of the
algorithm. To see this, let $R$ be the height-1 region that contains both $w_h$ and $u_h$. By Observation~\ref{obs:height0}, neither $w_h$ or $u_h$ are boundary vertices of $R$. Since the label of $u_h$ is smaller than that of $w_h$ no arc whose tail is $w_h$ would become the minimum of the heap $Q(R)$ as long as $w_h$ is latent. 
This implies, by a trivial inductive argument,   that  for any $i \geq 1$, for any height-$i$ region
$R$, the minimum elements in the heap $Q(R)$ in our algorithm and
in algorithm H have the same label, and in fact correspond to the same
arc (or to a hyperarc $e$ in our algorithm and to one of the arcs that
belong to $e$ in algorithm H). It follows that the two algorithms
perform the same relaxations and in the same order. Hence our
algorithm produces the same labels as algorithm H, and is therefore correct.

\subsection{Analysis}\label{sec:analysis}
The analysis follows that of HKRS~\cite{HKRS97}. We use the more recent
description of the charging scheme in~\cite{Book}, which
differs from the original one 
in~\cite{HKRS97} in the organization and presentation of the charging
scheme, but in essence is the same.

Since our algorithm only differs from that of~\cite{HKRS97} in the
implementation of height-0 regions, large parts of the analysis, and in
particular parts that do not rely on the choice of parameters (namely, 
the choice of region sizes $r_i$, and attention span $\alpha_i$)
remain valid without any change. The
most important of these is the Payoff Theorem (the charging scheme
invariant~\cite[Lemma 3.15]{HKRS97}), whose proof only depends on
processing elements with (locally) minimum labels, and on the fact
that processing a region $R$ only decreases labels of vertices that belong
to $R$.  
Since our analysis relies on the details of the original analysis, we
include some of the definitions and statement of lemmas
from~\cite{Book}.

We define an {\em entry vertex} of a region $R$ (denoted $\entry(R)$) as follows.  The only
entry vertex of the region $R_G$ (the region consisting of the entire graph $G$) is $s$ (the source) itself.  For any other region
$R$, $v$ is an entry vertex if $v$ is a boundary vertex of $R$.
We define the height of a vertex $v$ to be the largest integer $i$ such that $v$ is
an entry vertex of a height-$i$ region. 

When the algorithm processes a region $R$, it finds shorter paths to
some of the vertices of $R$ and so reduces their labels.  Suppose one
such vertex $v$ is a boundary vertex of $R$.  The result is that the
shorter path to $v$ can lead to shorter paths to vertices in a
neighboring region $R'$ for which $v$ is an entry vertex.  In order to preserve the property that
the minimum key of $Q(R')$ reflects the labels of vertices of $R'$ the
algorithm might need to update $Q(R')$.  Updating 
the queues of neighboring regions is handled by the {\sc Update}
procedure.  The reduction of $\minKey(Q(R'))$ (which can only
occur as a result of a reduction in the label of an entry vertex $v$ of $R'$)
is a highly significant event for the analysis.   We refer to such an
event as a {\em foreign intrusion of region $R'$ via entry vertex
  $v$.}

Because of the recursive structure of {\sc Process}, each initial
invocation {\sc Process}$(R_G)$ is the root of a tree of invocations
of {\sc Process} and {\sc Update}, each with an associated region $R$.  Parents,
children, ancestors, and descendants of recursive  invocations are defined in the
usual way.

For an invocation $A$ of {\sc Process} on region $R$, we define
$\pstart(A)$ and $\pend(A)$ to be the values of $\minKey(Q(R))$ just
before the invocation starts and just after the invocation ends,
respectively.

To facilitate the analysis we augment the pseudocode of {\sc Process}
and {\sc Update} to keep track of the costs of the various
operations. These additions to the pseudocode are shown in blue.
Note that these additions are purely
an expository device for the purpose of analysis; the additional blue code
is not intended to be actually executed. In fact,
Line~\ref{line:stable}  of {\sc Process} cannot be executed since
it requires knowledge of the future!

Amounts of cost are passed around by {\sc Update} and {\sc Process}
via return values and arguments.  We think of these amounts as {\em
  debt obligations}.  The running time of the algorithm is dominated
by the time for priority-queue operations and for operations on the
Monge Heaps. New debt is generated for every such call in
Lines~\ref{line:hyperdebt},~\ref{line:arcdebt}, and~\ref{line:proc}
of {\sc Process}, and in Lines~\ref{line:upd1} and~\ref{line:upd2} of
  {\sc Update}.
These debt obligations travel up and down the forest
of invocations. An invocation passes down some of its debt to its
children in hope they will pay this debt. A child, however, may pass
unpaid debt (inherited from the parent or generated by its
descendants) to its parent. Debts are eventually charged by invocations of {\sc
  Process} to pairs $(R, v)$ where $R$ is a region and $v$ is an
entry vertex of~$R$.

We say an invocation $A$ of {\sc Process} is {\em stable} if, for
every invocation $B>A$, the start key of $B$ is at least the start key
of $A$.
If an invocation is stable, it pays off the debt by
withdrawing the necessary amount from an account, the account associated with the
pair $(R,v)$ where $v$ is the value of $\entry(R)$ at the time of the
invocation. This value is set by {\sc Update} whenever a foreign
intrusion occurs (Line~\ref{line:entry}).  Initially the entries in the table are
undefined.  However, for any region $R$, the only way that
$\minKey(Q(R))$ can become finite is by an intrusion.  We are
therefore guaranteed that, at any time at which $\minKey(Q(R))$ is
finite, $\entry[R]$ is an entry vertex of~$R$.

Any invocation whose region is the whole graph is stable because there
are no foreign intrusions of that region.  Therefore, such an
invocation never tries to pass any debt to its nonexistent parent.  We
are therefore guaranteed that all costs incurred by the algorithm are
eventually charged to accounts.

 The following theorem,
called the Payoff Theorem, is the main element in the analysis. As we
already noted, the theorem, as well as its original proofs in~\cite{HKRS97,Book}, apply to
our algorithm.

\begin{theorem}[Payoff Theorem (Lemma 3.15 in~\cite{HKRS97})] \label{thm:Henzinger-payoff}
For each region $R$ and entry vertex $v$ of $R$, 
the account $(R, v)$ is used to pay off a positive amount at most once.
\end{theorem}

The remainder of the analysis consists of bounding the total debt
charged from all accounts. The original analysis in~\cite{HKRS97}
proved a bound on the debt that is linear in the number of arcs of the
planar graph $G$, which is also linear in the number of vertices of $G$ . 
In our case however, the number of arcs in the DDG is $O(n)$, but we
need a bound that is linear in the number of vertices of the DDG,
which is $O(n/\sqrt{r})$. 
To obtain a stronger bound we need to choose different parameters for
the sizes of regions in the $\rvec$--division than those used
in~\cite{HKRS97}. 
The analysis is rather technical so the details are deferred to 
Appendix~\ref{appendix:analysis}. The final bounds of our algorithm
are summarized by the following theorem.

\begin{theorem}\label{thm:analysis}
A shortest path computation on the dense distance graph of an
$r$--division of an $n$-vertex graph can be done in $O((n/\sqrt{r})\log^2 r)$ time after $O(n\log r)$ preprocessing, or in $O((n/\sqrt{r})\log^4 r)$ time after $O((n/\sqrt{r})\log^2 r)$ preprocessing.
\end{theorem}

\newcommand{\dual}[1]{#1^\ast}
\newcommand{\excess}{\mathit{excess}}
\renewcommand{\algorithmiccomment}[1]{\hfill{// #1}}

\section{Applications} \label{section:applications}

In this section we give two applications of our fast shortest-path algorithm. In both applications, we obtain a speedup over previous algorithms by decomposing a region of $n$ vertices using an $r$-division and computing distances among the boundary vertices of the region in $O((n/\sqrt{r})\log^2 r)$ time using our fast shortest-path algorithm.

\subsection{All-pair distances among boundary vertices of a small face}\label{sec:DDGapp}

Let $G$ be a directed planar graph, and let $\phi$ be a face of $G$ with $k$ vertices on its boundary. We consider the problem of computing the $O(k^2)$ distances among all $k$ boundary vertices of $\phi$. There are two previously known solutions for this problem. First, it is possible to compute the shortest-path tree from each of the $k$ boundary vertices. This solution takes $O(nk)$ time using the shortest-path algorithm of Henzinger et al.~\cite{HKRS97}. Second, we can compute the distances by applying Klein's multiple-source shortest-path algorithm (MSSP)~\cite{Klein05} in $O((n + k^2) \log n)$ time.

We use our fast shortest-path algorithm and give an algorithm that computes the all-pair distances among the $k$ boundary vertices in $O(n \log k)$ time, when $k < \sqrt{n} / \log n$. For larger values of $k$, note that the MSSP solution is upper bounded by $O((n + k^2) \log k)$ since $\log n = O(\log k)$. Therefore, we get an upper bound of $O((n + k^2) \log k)$ for computing the all-pair distances among the $k$ boundary vertices of a given face, for any value of $k$.

We consider the entire graph $G$ as a single region of an $n$-division of some other, larger graph (the large graph itself is not relevant for our algorithm). We define the face $\phi$ to be a hole of the region $G$, and the $k$ vertices of $\phi$ to be boundary vertices of the region $G$. Note that our definition of the boundary of $G$ is valid for a region in an $n$-division with a constant number of holes since we assume that $k < \sqrt{n}$.

We further decompose the region $G$ using an $r$-division, for some value of $r$ to be defined later. This decomposition maintains the property that the $k$ vertices of $\phi$ are boundary vertices of the regions of the $r$-division that contain them. We compute this $r$-division in linear time using the algorithm of Klein et al.~\cite{ShayrDivision}. Then, we compute the DDG of the $r$-division in $O(n \log r)$ time, using the MSSP algorithm of Klein~\cite{Klein05}.

The $k$ boundary vertices of $\phi$ are all vertices of the DDG. We apply our fast shortest-path algorithm from each of the $k$ vertices, and retrieve the required all-pair distances among them. Applying our algorithm $k$ times requires $O(k \cdot (n / \sqrt{r}) \log^2 r)$ time.

The total running time of our algorithm is $O(n \log r + k (n / \sqrt{r}) \log^2 r)$. We choose $r = k^2 \log^2 k$; note that since $k < \sqrt{n} / \log n$ we have that $r < n $, and an $r$-division of $G$ is indeed defined. We get that the running time of our algorithm is $O(n \log k)$, as required.

\subsection{Maximum flow when the source and the sink are close} \label{section:maxflow}
In this section, we use our fast shortest-path algorithm to show an $O(n \log p)$ maximum $st$-flow algorithm in directed planar graphs. The parameter $p$, first introduced by Itai and Shiloach~\cite{IS79}, is defined to be the minimum number of faces that a curve $C$ from $s$ to $t$ passes through. Note that the parameter $p$ depends on the specific embedding of the planar graph. 

If $p = 1$ then the graph is \emph{$st$-planar}. In this case, the maximum flow algorithm of Hassin~\cite{Hassin}, with the improvement of Henzinger et al.~\cite{HKRS97} runs in $O(n)$ time. For $p>1$,
Itai and Shiloach~\cite{IS79} gave an $O(n p \log n)$ time maximum flow algorithm when the value of the flow is known. Johnson and Venkatesan~\cite{JV83} obtained the same running time without knowing the flow value in advance. Using the shortest-path algorithm of Henzinger et al.~\cite{HKRS97}, it is possible to implement the algorithm of Johnson and Venkatesan in $O(n p)$ time. Borradaile and Harutyunyan~\cite{BH13} gave another $O(n p)$ time maximum flow algorithm for planar graphs. For undirected planar graphs, Kaplan and Nussbaum~\cite{KaplanN11} showed how to find the minimum $s-t$ cut in $O(n \log p)$ time.

Our $O(n \log p)$ time bound matches the fastest previously known maximum flow algorithms for $p = \Theta(1)$~\cite{JV83} and for $\log p = \Theta(\log n)$~\cite{borradaile-klein-09}, and is asymptotically faster than previous algorithms for other values of $p$. We assume that $p < \sqrt{n}/ \log^3 n$, for larger values of $p$ the $O(n\log n)$ maximum flow algorithm of Borradaile and Klein~\cite{borradaile-klein-09} already runs in $O(n \log p)$ time.

\paragraph{\bf Preliminaries.}
A planar flow network consists of: (1) a directed planar graph $G$; (2) two vertices $s$ and $t$ of $G$ designated as a \emph{source} and a \emph{sink}, respectively; and (3) a \emph{capacity} function $c$ defined over the arcs of $G$. We assume that for every arc $u v$, the arc $v u$ is also in the graph (if this is not the case, then we add $v u$ with capacity $0$), such that the two arcs are embedded in the plane as a single edge. A \emph{flow} function $f$ assigns a flow value to the arcs of $G$ such that: (1) for every arc $u v$, $f(vu) = -f(uv)$; (2) for every arc $u v$, $f(u v) \leq c(u v)$; and (3) for every vertex $v \neq s, t$, the amount of flow that enters $v$ equals to the amount of flow that leaves $v$. A \emph{maximum flow} is a flow in which the total amount of flow that enters the sink is maximal. In a \emph{preflow} some of the vertices of $G$ (in addition to $t$) may have more incoming flow than outgoing flow (an \emph{excess}). A \emph{maximum preflow} is a preflow in which the total amount of flow that enters the sink is maximal. Finally, we \emph{add} a flow $f'$ to a flow $f$ by setting $f(e) \leftarrow f(e) + f'(e)$ for every arc $e$.

The \emph{dual graph} $\dual{G}$ of $G$ is is a planar graph such that every face $\phi$ of $G$ has a vertex $\dual{\phi}$ in $\dual{G}$, and every arc $e$ of $G$ with a face $\phi_\ell$ to its left and a face $\phi_r$ to its right has an arc $\dual{e}=\dual{\phi_\ell} \dual{\phi_r}$ in $\dual{G}$.

We now describe our maximum flow algorithm. We first describe an $O(np)$-time algorithm and then show how to improve it to $O(n \log p)$ using ideas of Borradaile et al.~\cite{BKMNWN11} together with our fast shortest-path algorithm.

\paragraph{\bf An $O(np)$ algorithm.}
Without loss of generality, we may assume that the curve $C$ from $s$ to $t$ crosses $G$ only at vertices, say $v_0, v_1, v_2, \dots, v_p$, where $v_0 = s$ and $v_p = t$. Therefore, we can embed in the graph a path $P$ between $s$ and $t$ consisting of $p$ edges $v_{i-1} v_i$ for $1 \leq i \leq p$ (note that some of the edges of $P$ may be parallel to original edges of the graph). We set the capacities of the arcs of $P$ to $0$, so the value of the maximum flow is not affected by adding $P$ to the graph.

Intuitively, our algorithm iterates over the vertices of $P$, from $v_0 = s$ until the last vertex before $t$, $v_{p-1}$. At each iteration, the algorithm pushes forward the excess flow of the current vertex of $P$ to the vertices of $P$ that follow it, where the excess of $s$ is $\infty$, and the excess of any other vertex $v_i$ of $P$ is the excess flow remaining at $v_i$ following the previous iterations. This is a special case of an algorithm of Borradaile et al.\ for balancing flow excesses and flow deficits of vertices along a path~\cite{BKMNWN11}. We obtain a speedup over the algorithm of Borradaile et al.\ by using our fast shortest-path algorithm, in a way similar to the one in the previous section, for computing distances among boundary vertices of a face.

Our algorithm works as follows. For every arc $v_i v_{i-1}$, for $1 \leq i \leq p$, we set $c(v_i v_{i-1}) = \infty$. We initialize a flow $f$ by setting $f(e) = 0$ for every arc of the graph. Then, at iteration $i = 1, \dots p$, we send a flow from $v_{i-1}$ to $v_i$ in the residual graph of $f$. The flow we send is as large as possible, but not larger than the excess of $v_{i-1}$. For $i = 1$, the excess of $v_0 = s$ is $\infty$. For $i > 1$, the excess of $v_{i-1}$ is the value of the flow that we sent in the previous iteration from $v_{i-2}$ to $v_{i-1}$. 
In fact, at iteration $i$ we actually send as much flow as possible from $v_{i-1}$ to \emph{all} vertices of $P$ that follow it. This is because each such vertex $v_k$ is connected to $v_i$ with a path of infinite capacity so any flow that we can send to $v_k$ is routed to $v_i$. Note that each iteration is actually a flow computation in an $st$-planar graph, so we can implement in $O(n)$ time using Hassin's algorithm~\cite{Hassin,HKRS97}. After we compute the flow from $v_{i-1}$ to $v_i$, we add it to $f$ and continue to the next iteration.

Following the last iteration, $f$ is a maximum preflow from $s$ to $t$. Note that only vertices of $P$ may have excess. We discard the arcs of $P$ from the graph. Then, we convert the maximum preflow $f$ to a maximum flow using the following procedure of Johnson and Venkatesan. First, we make the flow acyclic using the algorithm of Kaplan and Nussbaum~\cite{KN11a} in $O(n)$ time. Then, we find a topological ordering of the vertices such that if an arc $uv$ carries a positive amount of flow, then $u$ precedes $v$ in the ordering. Finally, in $O(n)$ time, we return the excess from vertices with an excess flow to $s$ backwards along this topological ordering. 

Since each iteration takes $O(n)$ time, the total running time of the algorithm is $O(np)$. The bottleneck of the algorithm is applying Hassin's maximum flow algorithm for $st$-planar graphs $p$ times. We next use a technique of Borradaile et al.\ to execute Hassin's algorithm implicitly.

\paragraph{\bf An $O(n\log p)$ algorithm.}
We start by reviewing Hassin's maximum flow algorithm for $st$-planar graphs. The input for the algorithm is a flow network $G$ with a source $s$ and a sink $t$, such that $s$ and $t$ are on the boundary of a common face. We embed an arc $t s$ with infinite capacity in $G$. Let $\phi$ be the face of $G$ on the left-hand side of $t s$. We compute shortest path distances in $\dual{G}$ from $\dual{\phi}$, where the length of a dual arc equals the capacity of the corresponding primal arc. 

Let $\pi(\psi)$ denote the distance in $\dual{G}$ from $\dual{\phi}$ to $\dual{\psi}$. The function $\pi$, which we call a \emph{potential function}, defines a \emph{circulation} $f_\pi$ in $G$ as follows. For an arc $e$ with a face $\psi_\ell$ to its left and a face $\psi_r$ to its right, the flow along $e$ is $f_\pi(e) = \pi(\psi_r) - \pi(\psi_\ell)$. Hassin showed that after we drop the extra arc $t s$, we remain with a maximum flow $f_\pi$ from $s$ to $t$ in the original graph.

Following Borradaile et al.\ we apply Hassin's algorithm with two modifications. First, we compute a maximum flow with a prescribed bound $x$ on its value. We obtain such a flow by setting the capacity of $t s$ to $x$ rather than to $\infty$~\cite{KNK93}. Second, we do not wish to add and remove the arc $t s$, since we do not want to modify the embedding of the graph. In our case, the arc $t s$ already exists in the graph, this is the arc $v_i v_{i-1}$ with capacity $0$. Instead of adding a new arc, with capacity $x$, from the sink to the source, we set the capacity of $v_i v_{i-1}$ to $x$. Instead of removing the arc following the flow computation, we set the \emph{residual capacity} of $v_i v_{i-1}$ and $v_{i-1} v_i$ back to $0$, by setting their capacities to $f(v_i v_{i-1})$ and $-f(v_i v_{i-1})$, respectively.

We do not maintain a flow $f$ throughout our algorithm, we rather use the potential function $\pi$ to represent the flow $f_\pi$. This gives us the implementation of our maximum flow algorithm in Algorithm~\ref{alg:flow} below.

\newcommand{\green}[1]{\textcolor{Green}{#1}}
\begin{algorithm} \caption{{\sc MaximumFlow}$(G, s, t, c)$} \label{alg:flow}
  \begin{algorithmic}[1]
	\STATE Embed the path $P = (v_0, v_1, \dots, v_p)$, where $v_0 = s$ and $v_p = t$, with $c(v_i v_{i-1}) = \infty$ and $c(v_{i-1} v_i) = 0$ for every $1 \leq i \leq p$
	\STATE $\pi(\psi) \leftarrow 0$ for every face $\psi$ of $G$ \green{\COMMENT{implicitly initialize $f_\pi(e) = 0$ for every arc $e$}}
	\STATE $\excess \leftarrow \infty$ \green{\COMMENT{the excess of $s$ is $\infty$}}
	\FOR {$i = 1$ to $p$}
		\STATE $c(v_i v_{i-1}) \leftarrow \excess$ \green{\COMMENT{bounds the amount of flow we send from $v_{i-1}$ to $v_i$}}
		\STATE $\pi_i(\psi) \leftarrow$ the distance from $\dual{\phi_i}$ to $\dual{\psi}$ in $\dual{G}$, for every face $\psi$, where $\phi_i$ is the face to the left of $v_i v_{i-1}$, and for every arc $e$ the length of $\dual{e}$ is the residual capacity of $e$ with respect to $f_\pi$ \green{\COMMENT{implements Hassin's maximum flow algorithm}}\label{line:flow-sp}
		\STATE $\pi(\psi) \leftarrow \pi(\psi) + \pi_i(\psi)$, for every face $\psi$ \green{\COMMENT{adds the new flow to $f_\pi$}}\label{line:accumulate}
		\STATE $c(v_i v_{i-1}) \leftarrow f_\pi(v_i v_{i-1})$; $c(v_{i-1}, v_i) \leftarrow -f_\pi(v_i v_{i-1})$ \green{\COMMENT{set the residual capacity of both arcs to $0$}}\label{line:reweight}
		\STATE $\excess \leftarrow f_\pi(v_i v_{i-1})$ \green{\COMMENT{this is the value of the flow}}
	\ENDFOR
	\STATE Discard the path $P$ from the graph \label{line:after-for}
	\STATE Convert the maximum preflow $f_\pi$ into a maximum flow $f$ \green{\COMMENT{the implementation of this step is described above}}
	\RETURN $f$
  \end{algorithmic}
\end{algorithm}

If we implement the shortest path computation at Line~\ref{line:flow-sp} using the algorithm of Henzinger et al., then the running time of our algorithm is $O(np)$. Borradaile et al.\ noticed that we do not have to compute and maintain the value of the potential function $\pi$ for all faces of $G$ prior to Line~\ref{line:after-for}. Until this point, it suffices to compute the value of the potential function only for the faces incident to the arcs of $P$. This way, Borradaile et al.\ computed the distances in $\dual{G}$ using FR-Dijkstra on a DDG that contains the dual vertices of all of these faces, and executed Line~\ref{line:flow-sp} faster. We follow this approach, and compute the distances at Line~\ref{line:flow-sp} in a way similar to the algorithm of the previous section. Only after the last iteration of the loop that computes $\pi$, we compute explicitly the value of $\pi$ for all faces of the graph. We now give the details of this improvement.

Consider the graph $\dual{G}$, the dual graph of $G$ (including the arcs of $P$). Let $\dual{X}$ be the set of vertices of $\dual{G}$ that correspond to faces of $G$ incident to arcs of $P$. The set $\dual{X}$ is the set of vertices for which we compute distances at Line~\ref{line:flow-sp}. We define the region $\dual{R}$ to be the region of $\dual{G}$ that contains the dual arcs of the original arcs of $G$, excluding the arcs of $P$. The boundary vertices of $\dual{R}$ are exactly the vertices of $\dual{X}$, and they all lie on a single face of $\dual{R}$ (which may have other vertices on its boundary). Let $\dual{P}$ be the region of $\dual{G}$ that contains the arcs dual to arcs of $P$, these are the arcs of $\dual{G}$ that are not in $\dual{R}$. See Figure~\ref{fig:flownetwork}.

\begin{figure}[tb]
	\begin{center}
		\includegraphics[scale=0.8]{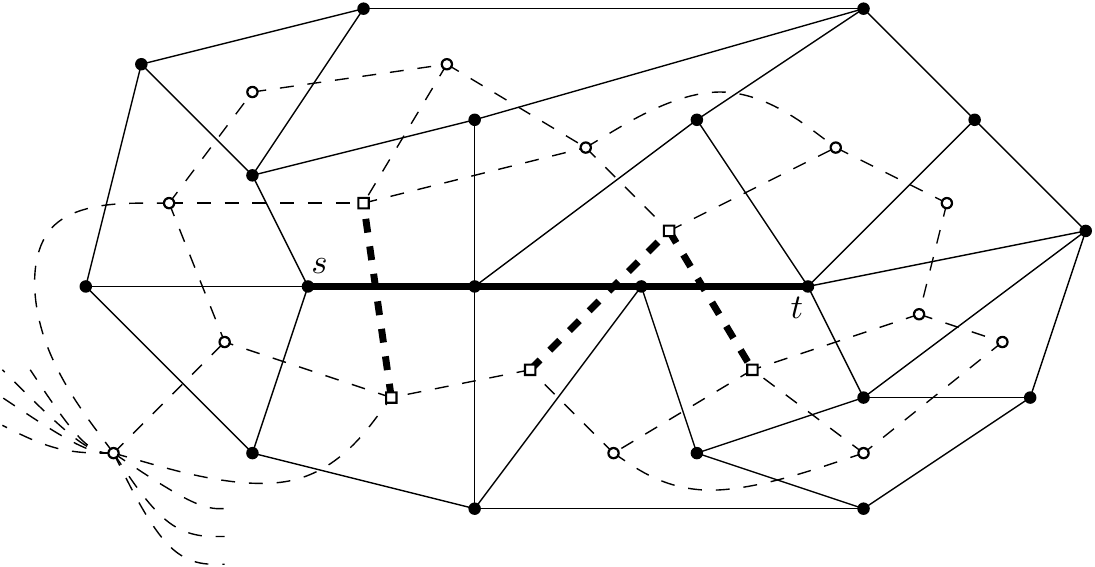}
		\caption{A flow network $G$ (\emph{solid} edges) and the dual network $\dual{G}$ (\emph{dashed} edges). The path $P$ from $s$ to $t$ is \emph{bold}, the region $\dual{P}$ is \emph{dashed and bold}. The set of dual vertices $\dual{X}$ (\emph{boxes}) separates between $\dual{P}$ and $\dual{R}$ (the dual edges that are not in $\dual{P}$), they are all on the boundary of a single face of $\dual{R}$. To make the illustration simpler, some of the arcs incident to the dual vertex of the infinite face are omitted.}
		\label{fig:flownetwork}
	\end{center}
\end{figure}

We decompose $\dual{R}$ using an $r$-division, similar to the way we did in the previous section, such that the vertices of $\dual{X}$ remain boundary vertices of every region containing them. We choose $r = p^2 \log^6 p$ for reasons that will become clear later. Note that $r < n$, so an $r$-division is well defined. The region $\dual{P}$ has $O(r)$ vertices and $O(\sqrt{r})$ boundary vertices (actually it has much less vertices than $r$ and much less boundary vertices than $\sqrt{r}$). Therefore, the regions of the $r$-division of $\dual{R}$ together with the region $\dual{P}$ are an $r$-division of $\dual{G}$, in which all vertices of $\dual{X}$ are boundary vertices. We compute a DDG for this $r$-division in $O(n \log r)$ time.

Initially, the lengths of the arcs of $\dual{R}$ are defined by the capacities of the corresponding primal arcs. At each invocation of Line~\ref{line:accumulate}, we update the potential function $\pi$, and thereby also the residual capacities of the arcs of $G$. Since the length of a dual arc equals to the residual capacity of the corresponding primal arc, we update the lengths of the arcs of $R$ as well. As Borradaile et al.\ noticed, we can use $\pi$ also as a \emph{price function} that represents the change in the lengths of the arcs of $\dual{R}$ and defines \emph{reduced lengths}. That is, if the distance from $\dual{\psi_x} \in \dual{X}$ to $\dual{\psi_y} \in \dual{X}$ with respect to the original capacity of the graph is $d$, then the distance from $\dual{\psi_x}$ to $\dual{\psi_y}$ with respect to the residual capacities of $f_\pi$ is $d - (\pi(\psi_y) - \pi(\psi_x))$. We conclude that we compute the DDGs of the regions that compose $\dual{R}$ only once and use $\pi$ as a price function for these DDGs.

At Line~\ref{line:reweight} of the algorithm, we update the lengths of $v_i v_{i-1}$ and $v_{i-1} v_i$. This update causes a change in the DDG of $\dual{P}$. Since $\dual{P}$ contains only $p$ vertices, we can recompute the DDG in $O(p^2 \log p) = O(r)$ time. 

We implement Line~\ref{line:flow-sp} by applying our fast
shortest-path algorithm to the DDG of $\dual{G}$. We use the variant
of our algorithm that allows a price function and supports reduced
lengths (i.e, the variant with $c_q=3$ whose preprocessing time is $O((n/\sqrt{r})
\log^2r)$ running time is $O((n/\sqrt{r}) \log^4r)$, see
Theorem~\ref{thm:analysis} and Section~\ref{section:MH-spec}).  Thus, the total time for running the loop is $O(p(n/\sqrt{r}) \log^4 r)$. Since we set $r = p^2 \log^6 p$, the running time of the loop is $O(n \log p)$.

It remains to show how to obtain the value of the potential function
$\pi$ for all faces of $G$ prior to Line~\ref{line:after-for}. Recall
that, for a face $\psi$ of $G$, $\pi(\psi)$ is the distance in
$\dual{G}$ from the dual vertex of the face $\phi_p$ to the left of
$v_p v_{p-1}$ to $\dual{\psi}$, where the length of each dual arc is
defined by the capacity of the corresponding primal arc. Therefore, we
can compute $\pi$ using a shortest-path algorithm. We do not apply a
shortest-path algorithm for the entire graph $\dual{G}$, since the
arcs dual to arcs of $P$ may have negative lengths (which set the
residual capacities of the primal arcs with respect to $f_\pi$ to
$0$). Instead, we initialize the labels of the vertices of $\dual{X}$
according to the values of $\pi$ that we have already computed, and
apply the shortest-path algorithm of Henzinger et al.\ to $\dual{R}$,
which contains only non-negative length arcs. This is analogous to the
way Borradaile et al.\ handle the same issue.

We conclude that the total running time of our algorithm is $O(n \log p)$ as required. The correctness of our algorithm follows from the correctness of the algorithm of Borradaile et al., since our algorithm is a special case of it. The differences between our algorithm and the one of Borradaile et al.\ are that we decompose the region $\dual{R}$ using an $r$-division and that we apply our fast shortest-path algorithm.

\section {Acknowledgements}
We thank Philip Klein for discussions and for allowing us to use large
parts of his description~\cite{Book} of the algorithm of
Henzinger et al.~\cite{HKRS97} as the basis for our
description in section~\ref{sec:HKRS-FR}.

\appendix

\section*{Appendix}
\section{Proof of Lemma~\ref{lemma:RMQ}.} \label{ap:RMQproof}
First, as we show in~\cite{GMW}: The blank entries in a partial Monge matrix $M$ can be implicitly replaced so that $M$ becomes fully Monge and each entry $M_{ij}$ can be returned in $O(1)$ time. We can therefore assume that $M$ is fully Monge. To make the presentation clear (and compatible with~\cite{KaplanMNS12}), we prove the lemma on the transpose of $M$ (i.e., we activate and deactivate rows and a query is a column and a range of rows). Furthermore, 
although the lemma states that $M$ is an $n \times n$ matrix, we consider $M$ an an $m \times n$ matrix for some $m\le n$.
  
Denote $r(j)=i$ if the  minimum element of $M$ in
column $j$ lies in row $i$.  
The {\em upper envelope} $\mathcal{E}$ of all the rows of the Monge matrix
$M$ consists of the $n$ values $r(1), \ldots, r(n)$. Since $M$ is Monge we have that $r(1)\le  \ldots \le r(n)$ and so $\mathcal{E}$ can be implicitly represented in $O(m)$ space by keeping only the $r(j)$s of $O(m)$ columns called {\em breakpoints}. Breakpoints are the columns~$j$ where $r(j)\ne r(j+1)$. 
The minimum element $r(\pi)$ of an entire column $\pi$ can then be retrieved
in $O(\log m)$ time by  a binary search for the first
breakpoint column $j$ after $\pi$, and  setting $r(\pi)=r(j)$.

The tree $T_h$ presented in~\cite{KaplanMNS12} is a full binary
tree $T_h$ whose leaves are the rows of $M$. A node $u$ whose subtree contains $k$
leaves (i.e., $k$ rows) stores the $O(k)$ breakpoints of the $k \times
n$ matrix $M_u$ defined by these $k$ rows and all $n$ columns of $M$ (each breakpoint also stores its appropriate row index).
A leaf represents a single row and
requires no computation.
For an internal node $u$, we compute its list of breakpoints  by merging the breakpoint lists of its  left child
$u_\ell$ and  its right child $u_r$, where $u_\ell$ is the child whose rows have lower indices.

By the Monge property, $u$'s list of
 breakpoints starts with a prefix of  $u_\ell$'s breakpoints and ends with a suffix of $u_r$'s breakpoints. 
Between these there is possibly one new breakpoint $j$ (the {\em transition} breakpoint). The prefix and suffix parts can be found easily in $O(k)$ time by linearly comparing 
the relative order of the upper envelopes of $M_{u_\ell}$ and $M_{u_r}$.
This allows us to find an interval of columns containing the transition column $j$. Over this interval, the transition row is in one of two rows (a row in $M_{u_\ell}$ and a row in  $M_{u_r}$), and so $j$ can be found in  $O(\log n)$ time via binary search. Summing $O(k + \log n)$ over all 
nodes
of $T_h$ gives $O(n \log n)$  time. 
The
total size of $T_h$ is $O(m \log m)$. A query to $T_h$ is performed as follows.
 
The minimum entry in a query column $\pi$ and a contiguous range of rows $R$ is found by identifying $O(\log m)$ \emph{canonical nodes} of $T_h$.
A node $u$ is canonical if $u$'s set of rows is contained in $R$ but the set of
rows of $u$'s parent is not.
For each such canonical node $u$, the minimum element in column $\pi$ amongst all the rows of $u$ is found in $O(\log m)$ time via binary search on $u$'s list of breakpoints.
The output is the smallest of these. The  total query time is thus
$O(\log^2 m)$.

\paragraph{\bf A dynamic data structure.} The above $T_h$ data structure was used in~\cite{KaplanMNS12} in a static version and so the query time was  improved from $O(\log^2 m)$ to $O(\log m)$  using fractional cascading~\cite{CGfrac}. In order to allow activating/deactivating of a row we slightly modify this data structure and do not use fractional cascading. 
Every node $u$ of  $T_h$, instead of storing all its breakpoints, will now only store its single transition breakpoint. This way, in order to find the minimum element in column $\pi$ of $M_u$, we need to find the predecessor breakpoint of $\pi$ in $M_u$. This is done in $O(\log m)$ time by following a path in $u$'s subtree starting from $u$ and ending in a leaf. Given a query column and a contiguous range of rows we again identify $O(\log m)$ canonical nodes. Each of them now requires following a path of length $O(\log m)$ for a total of $O(\log^2 m)$ time.

In order to deactivate an entire row $r$ ,we mark this row's  leaf in $T_h$ as  inactive and then fix the breakpoints bottom-up on the entire path from $r$ to the root of $T_h$. For each node $u$ on this path, we assume that 
the transition breakpoints for all descendants of $u$ are correct and we find the new transition breakpoint $j$ of $u$ as follows:  If $u$ has left child $u_\ell$ and right child $u_r$ then $j$ should be placed between a breakpoint $j_\ell$ of $M_{u_\ell}$ and a breakpoint $j_r$ of $M_{u_r}$. Here, $M_{u_\ell}$ and $M_{u_r}$  correspond  to the {\em active} rows in $u_\ell$'s and $u_r$'s subtrees. If there are no active rows in $u_\ell$'s or $u_r$'s subtrees then finding $j$ is trivial. If both have active rows then we now show how to find $j_\ell$, finding $j_r$ is done similarly. 

To find $j_\ell$, we follow a path from $u_\ell$ downwards until we reach a leaf or a node whose subtree contains no active rows as follows. Recall that the node $u_\ell$ stores a single breakpoint (a row index $i'$ and a column index $j'$). 
Since $(i',j')$ is the minimum entry in column $j'$ of $M_{u_\ell}$ we want to compare it with $(i'',j')$ = the minimum entry in column $j'$ of $M_{u_r}$.  If $M[i',j']<M[i'',j']$ then we know that $j>j'$ and so we can proceed from $u_\ell$ to $u_\ell$'s right child. Similarly, if $M[i',j']\ge M[i'',j']$ then $j\le j'$ and we  proceed from $u_\ell$ to $u_\ell$'s left child. The problem is we do not know $i''$. To find it, we need to find the minimum entry in column $j'$ of $M_{u_r}$. We have already seen how to do this in $O(\log m)$ time by following a downward path from $u_r$. We therefore get that finding $j_\ell$ (and similarly $j_r$) can be done in $O(\log^2 m)$ time. 
After finding $j_\ell$ and $j_r$ we  find where $j$ is between them in $O(\log n)$ time via binary search. The overall time for deactivating an entire row $r$  is thus $O(\log^3 m + \log m \log n) =O(\log^3 n)$.

In order to activate an entire row $r$ we mark its leaf in  $T_h$ as active and fix the breakpoints from $r$ upwards using a similar process as above. This completes the proof of the dynamic data structure. 

\paragraph{\bf A decremental data structure.} 
We now proceed to proving the decremental data structure in which inactive rows never become active. 
We  take advantage of this fact  by maintaining, for each node $u$ of $T_h$, not only its transition breakpoint but all the active breakpoints of $M_u$.  
Since breakpoints are integers bounded by $n$, we can maintain $u$'s $k$ active breakpoints in an $O(k)$-space $O(\log\log n)$ amortized query-time predecessor data structure such as Y-fast-tries~\cite{Willard1983}. Given a query column and a contiguous range of rows, we identify $O(\log m)$ canonical nodes in $T_h$  and then only need to do a predecessor search on each of them in total $O(\log m \log\log n)$ time. 

In order to deactivate a row $r$, we mark $r$'s leaf in $T_h$ as inactive and fix the predecessor data structures of $r$'s ancestors. 
For each node $u$ on the $r$-to-root path we may need to change $u$'s transition breakpoint as well as insert (possibly many) new breakpoints into $u$'s predecessor data structure $Y_u$.
We show how to do this assuming $u$'s left child $u_\ell$ is the ancestor of $r$ and we've already fixed its predecessor data structure $Y_\ell$. The case where $u$'s right child $u_r$ is the ancestor of $r$ is handled symmetrically.  
Note that, if $r$ is a descendent of $u_\ell$, then the transition breakpoint and the predecessor data structure $Y_r$ of $u_r$ undergo no changes.

The first thing we do is we check if $u$'s old transition breakpoint has row index $r$. If not then we don't need to change $u$'s transition breakpoint. If it is, then we delete $u$'s old transition breakpoint from $Y_u$ and we find the new transition breakpoint $j$ of $u$  via binary search. The binary search takes $O(\log n \log\log n)$ time since each comparison requires to query $Y_\ell$ and  $Y_r$ in $O(\log\log n)$ time. 
After finding the new transition breakpoint $j$ we insert it into $Y_u$. Finally, we need to insert to $Y_u$ (possibly many) new breakpoints, each in $O(\log\log n)$ time. These breakpoints are precisely all the breakpoints in $Y_\ell$ that are to the left of $j$ and to the right of predecessor($j$) in $Y_u$ as well as all the breakpoints in $Y_r$ that are to the right of $j$ and to the left of successor($j$) in $Y_u$.

We claim that the above procedure requires  $O(\log^2 n \log\log n)$ amortized time. First note that finding the new transition breakpoints of all $O(\log m)$ nodes on the $r$-to-root path takes $O(\log m \log n \log\log n)$ time. Inserting these new breakpoints into the predecessor data structures and deleting the old one takes $O(\log m \log\log n)$ time. 
Next, we need to bound the overall number of breakpoints that are ever inserted into predecessor data structures. In the beginning, all rows are active and the total number of breakpoints in all predecessor data structures is $O(m\log m)$. Then, whenever we deactivate a row, each node $u$ on the leaf-to-root path in $T_h$ may insert a new transition breakpoint $j$ into $Y_u$. This breakpoint $j$ might consequently be inserted into every $Y_v$ where $v$ is an ancestor of $u$. Let $top(j)$ denote the depth of the highest ancestor of $u$ to which we will consequently insert the breakpoint $j$. Note that we might also insert to $Y_u$ other breakpoints but that only means that their $top$ value would decrease. Since every $top$ value is at most $O(\log m)$ and since top values only ever decrease this means that over all the $m$ row deactivations only $O(m\log^2 m)$ breakpoints are ever inserted. This means that all insertions take $O(m\log^2 m\log\log n)$ time.
\qed

\section{Proof of Lemma~\ref{lem:HKRS-RMQ}}\label{ap:MHimplement}
In this section we describe two implementations of the Monge Heaps
specified in Section~\ref{section:MH-spec}, thus proving
Lemma~\ref{lem:HKRS-RMQ}.

\paragraph{Implementation 1}
This implementation in the statement of the lemma
is  similar to that of section~\ref{warmup} and differs only in the use of RMQ. In section~\ref{warmup} we used one RMQ data structure for all MHs of a region $R$. Here, each MH requires an RMQ data structure of its own. This means that for $i=1,\ldots, \log(\sqrt{r})$, a region $R$ requires $2^i$ RMQ data structures for $\sqrt{r}/2^i \times \sqrt{r}/2^i$ matrices. 
Furthermore, in section~\ref{warmup} we only deactivated vertices and so we could use the 
decremental RMQ of Lemma~\ref{lemma:RMQ}. Here, an inactive vertex can become active again and so we use the (slower) dynamic RMQ of Lemma~\ref{lemma:RMQ}. 

To implement {\sc FR-GetMinChild}$(h, v)$ we simply query the RMQ data structure of $h$ in $O(\log^2 r)$ time. 
To implement {\sc FR-Extract}$(v_h)$ we first deactivate $v_h$ in the RMQ  of $h$ in $O(\log^3 r)$ time (notice that the deactivation is done only in $h$ as opposed to section~\ref{warmup}  where a deactivation applies to all MHs that contain the vertex). 
Then, we find $w_h$ by querying the RMQ of $h$ in $O(\log^2 r)$ time. 
To implement {\sc FR-Relax}$(h, v)$ we first perform the same
procedure as in section~\ref{warmup} that creates one new triplet and
destroys possibly many triplets. This takes amortized $O(\log r)$ time
since destroying a triplet is paid for by that triplet's creation
(note that as opposed to section~\ref{warmup} here a triplet of vertex
$v$ can be destroyed and then created again multiple times). Then,
activating all the new children of $v$ that were inactive takes
amortized $O(1)$ time since we can charge the activation of a vertex
to the time it was deactivated (in a prior call to  {\sc
  FR-Extract}). 

Constructing the dynamic RMQ  of Lemma~\ref{lemma:RMQ} for an $x\times
x$ matrix takes $O(x \log x)$ time. Therefore all RMQ data structures
of a region $R$ are constructed in $\sum_{i=1}^{\log(\sqrt{r})} 2^i
\cdot \sqrt{r}/2^i \cdot \log(\sqrt{r}/2^i)=O(\sqrt{r} \log^2 r)$
time, and so over all regions we need $O((n/\sqrt{r}) \log^2 r)$
time. 

\paragraph{Implementation 2}
The second implementation in the statement of Lemma~\ref{lem:HKRS-RMQ}
is based on the following simple RMQ.

\begin{observation}
Given an $n\times n$ Monge matrix $M$, one can construct in $O(n^2\log n)$ time and $O(n^2)$ space a data structure that supports the following operations in $O(\log n)$ time:  (1) mark an entry as inactive (2) mark an entry as active (3) report the minimum active entry in a query
row and a contiguous range of columns (or $\infty$ if no such entry exists).
\end{observation}
\begin{proof}
The data structure is a very na\"ive one. We simply store, for each
row of $M$, a separate balanced binary search tree keyed by column
number where each node stores the minimum value in its
subtree. Deactivating (resp. activating) an entry is done by deletion
(resp. insertion) into the tree, and a query is done by taking the
minimum value of $O(\log n)$ canonical nodes identified by the query's
endpoints. \qed 
\end{proof}

Notice that the above na\"ive RMQ data structure activates and
deactivates single entries in $M$ while the specification of {\sc FR-Extract}
seems to require deactivating and not entire columns of $M$ (The
specification is to Make $v_h$ inactive in the entire MH $h$).
However, deactivating a single element suffices. 
This is because each
vertex $v_h \in B_h$ appears in exactly one triplet at any given
time. Therefore, whether it is active or inactive is only important in
the RMQ of the vertex $u \in A_h$ whose triplet contains $v_h$.
Whenever an inactive $v_h$ changes parents, we make $v_h$ active by reintroducing
the corresponding element into the RMQ if it was previously deleted,
and charge the activation 
to the time $v_h$ was deactivated in a prior call to  {\sc
  FR-Extract}). 

Constructing the na\"ive data structure on all MHs of
a region $R$ takes $\sum_{i=1}^{\log(\sqrt{r})} 2^i (\sqrt{r}/2^i)^2
\cdot \log(\sqrt{r}/2^i)  =O(r \log r)$ time, and so over all regions
we need $O(n\log r )$ time. Each operation on the RMQ now takes only
$O(\log r)$ time. This means that {\sc FR-Relax}, {\sc
  FR-GetMinChild}, and {\sc FR-Extract} 
can all be done in $O(\log r)$ amortized time. 
\qed

\section{Analysis of the Main Algorithm}\label{appendix:analysis}
In this section we give the complete running-time analysis of our algorithm by bounding the total debt withdrawn from all accounts during the execution of the algorithm.

The total debt depends on the parameters $\rvec =(r_1, r_2, \ldots)$ of the recursive $\rvec$-division and on the
parameters $\alpha_0, \alpha_1, \ldots$ that govern the number of
iterations per invocation of {\sc Process}.  Recall that $r_1 = r$  and
$r_i = r_{i-1}^2$ for   $i>1$.  Define $\alpha_0 = 1$ and $\alpha_i=\frac{4\log r_{i+1}}{3\log r_i}$ for $i>0$.

\begin{lemma} \label{lem:debt-inherited}
Each invocation at height $i$ inherits at most $4\log r_{i+1}$ debt.
\end{lemma}

\begin{proof} The proof is by reverse induction on $i$. In the base case, each top-height invocation
 clearly inherits no debt at all.  Suppose the lemma holds for $i$, and consider a
  height-$i$ invocation on a region $R$.  By the inductive hypothesis, this invocation inherits at
  most $4 \log r_{i+1}$ debt. An ${\alpha_i}$ fraction of this debt is
  passed down to each child. In addition, each child $R'$
  inherits a debt of $\log(|Q(R)|) = \log(r_i)$ to cover the cost of
  the $\minItem$ and $\increaseKey$ operations on $Q(R)$ required for
  invoking the call to {\sc Process} on $R'$. Overall we get that each child inherits  $\frac{4 \log r_{i+1}}{\alpha_i}+\log(r_i)$ which by 
the choice of $\alpha_i$ is $4\log r_i$. \qed
\end{proof}

We next wish to bound the number of descendant invocations incurred by an invocation of {\sc Process}.   
Define $\beta_{ij}=\alpha_i \alpha_{i-1} \ldots \alpha_{j+1}$ for
$i > j$. Define $\beta_{ij}$ to be 1 for $i=j$, and 0 for $i<j$. By definition, 
$$\beta_{ij} = \prod_{k=j+1}^i \frac{4\log r_{k+1}}{3\log r_k} =
\left( \frac{4}{3} \right)^{i-j} \frac{\log r_{i+1}}{\log r_{j+1}}.$$

\begin{lemma} \label{lem:beta}
A {\sc Process} invocation of height $i$ has at most
  $\beta_{ij}$ descendant {\sc Process} invocations of height $j$.
\end{lemma}

Note that only height-0 invocations incur debt. Let $c_3$ be a constant such that  a Monge heap operation takes $c_3\log^{c_q} r_1$ time.  The debt incurred by the algorithm by a step of {\sc Process} is called
the {\em process debt}. 

\begin{lemma} \label{lem:process-debt}
For each height-$i$ region $R$ and boundary vertex $x$ of $R$, the
amount of process debt payed off by the account $(R,x)$ is 
at most $c_3\beta_{i0}\log^{c_q} r_1  + \sum_{ j \leq i} 4\beta_{ij} \log r_{j+1} $.
\end{lemma}

\begin{proof} By the Payoff Theorem, the  $(R,x)$ account is used at most
  once.  Let $A$ be the invocation of {\sc Process} that withdraws the
  payoff from that account.  Each dollar of process debt paid off by
  $A$ was sent back to $A$ from some descendant invocation who
  inherited or incurred that dollar of debt. Since only height-0
  invocations incur debt, to account for the
  total amount of process debt paid off by $A$ we consider each of
  its descendant invocations.  By Lemma~\ref{lem:beta}, $A$ has
  $\beta_{ij}$ descendants of height~$j$, each of which inherited a debt of at
most $4 \log r_{j+1}$ dollars, by  Lemma~\ref{lem:debt-inherited}.
In addition, each of the $\beta_{i0}$ descendant height-0 invocations 
incurs a debt of at most $c_3\log^{c_q} r_1$. \qed
\end{proof}

Recall that $c_1, c_2$ are the constants such that an $r$-division
of an $n$-vertex graph has at most $c_1 n/r$ regions, each having at most
$c_2 \sqrt{r}$ boundary vertices.

\begin{lemma} \label{lem:number-of-region-vertex-pairs}
Let $n$ be the number of edges of the input graph.
For any nonnegative integer $i$, there are at most $c_1 c_2
n/\sqrt{r_i}$ pairs $(R,x)$ where $R$ is a height-$i$ region and $x$
is an entry vertex of $R$.
\end{lemma}

Combining this lemma with Lemma~\ref{lem:process-debt}, we obtain

\begin{corollary} \label{cor:total-process-debt}
The total process debt is at most 
$$c_1 c_2 \sum_i \frac{n}{\sqrt{r_i}} \left( c_3\beta_{i0}\log^{c_q} r_1
  + \sum_{j \leq i} 4\beta_{ij} \log r_{j+1}\right)
$$
\end{corollary}

The cost incurred by the algorithm by a step of {\sc Update} is called
{\em update debt}.  The event (in Lines~\ref{line:minchild-relax},~\ref{line:v-relax}, and~\ref{line:minchild_a} of {\sc Process}) of reducing
a vertex $x$'s label $d(x)$ initiates a chain of calls to {\sc Update}
in Lines~\ref{line:minchild-update}, \ref{line:update},
and~\ref{line:minchild_b}, respectively, for each outgoing arc $xy$.  We say the debt incurred is
{\em on behalf of $x$}.

\begin{lemma} \label{lem:update-debt}
If {\sc Update} is called on the parent of region $R$ during an invocation $A$
of {\sc Process} then $R$ is not the region of $A$. 
\end{lemma}

Recall that for a vertex $v$, we defined 
$$\text{height}(v) = \max\set{j\, : \, v \text{ is a boundary vertex of a height-$j$ region}}$$

\begin{lemma} \label{cor:cost-per-update-chain}
A chain of calls to {\sc Update} initiated by the
reduction of the label of a natural vertex $v$ has total cost $O(1+ \text{height}(v))$.

A chain of calls to {\sc Update} initiated by the
reduction of the label of a copy vertex $w_h$ has total cost $O(1)$.
\end{lemma}

\begin{proof}
 
Let $A_0$ be the invocation of Process during which the
initial call to {\sc Update} was made, and let $R$ be the (height-0)
region of $A_0$.  

For the first case, the call to {\sc Update} is in
line~\ref{line:update}.
In this case, $R$ corresponds to a single arc $v_hv$.
Consider the chain of calls to {\sc Update}, and let
$R_0, R_1, \ldots, R_p$
be the corresponding regions. 
$R_0$ corresponds to a single hyperarc $e$ whose tail is $v$.
Note that $\text{height}(R_j)=j$.  The cost
of each call is $O(1)$ since we use Fibonacci heaps.  Since
$R_{p-1}$ contains $e$ but (by Lemma~\ref{lem:update-debt}) does not
contain $v_hv$, we get that $v$ is a boundary vertex of $R_{p-1}$, so $p \leq
\text{height}(v)+1$. 

For the second case, the call is {\sc Update}$(R',w_hw,d(w_h))$, either in 
Line~\ref{line:minchild-update}, or in Line~\ref{line:minchild_b}.
In the former case $R$ corresponds to a single hyperarc incident to
$w_h$. 
In the latter case $R$ corresponds to a single arc $v_hv$.
In both cases $\parent(R)$ and $\parent(R')$ is the same height-1
region $R_1$, and the key of $R'$ is at least the key of $R$. Since at
the time the call to {\sc Update} is made $\minKey(Q(R_1)) = R$, 
the call {\sc Update}$(R',w_hw,d(w_h))$ does not decrease
$\minKey(Q(R_1))$, so this chain of calls to {\sc Update} consists of
exactly two calls. Since we use Fibonacci heaps, each call costs $O(1)$.
\qed
\end{proof}

\begin{lemma} \label{lem:update-debt-bound}
Let $i, j$ be nonnegative integers, and let $R$ be a region of height $i$.
The total amount of 
update debt incurred on behalf of vertices of height at most $j$  and 
paid off from accounts $\set{(R,x)\, :\, x \text{ is an entry vertex of }R}$ is at
most
$$3c_2 \sqrt{r_i} \beta_{i0} (3c_5\log r_1 + j).$$
\end{lemma}

\begin{proof}  The number of entry vertices $x$ of $R$ is at most $c_2
  \sqrt{r_i}$.  For each, the Payoff Theorem ensures that all the debt
  paid off from account $(R,x)$ comes from descendants of a single
  invocation $A$ of {\sc Process}.  The number of height-0 descendants
  of $A$ is $\beta_{i0}$.  
Consider such a height-0 descendant $A_0$ whose region is $R_0$.
If $R_0$ consists of a single hyperarc, then the height of the vertex $w_h$ on whose
behalf the call to {\sc Update} is made in
Line~\ref{line:minchild-update} is 0 (see Observation~\ref{obs:height0}). 
By Lemma~\ref{cor:cost-per-update-chain} the cost of this chain of
calls is at most $2$.
If $R_0$ consists of a single arc $v_hv$, then $v$ has height at least
1. Suppose the height of $v$
is at most $j$. In each of the three height-1 regions to which $v$ belongs, it
participates in at most $c_5\log r_1$ Monge heaps. Hence, in each of these
regions there are at most
$c_5\log r_1$ hyperarcs whose tail is $v$ for which a call to {\sc
  Update} is made in Line~\ref{line:update}.  Since all of these calls
are made with the same key, exactly one of them may cause a chain of
update whose length is at most $j$. By Lemma~\ref{cor:cost-per-update-chain} the cost
  of such a chain is at most $j+2$. The other calls only cause a chain of
updates whose length is 2. Hence, all calls made in
Line~\ref{line:update} on behalf of $v$ cost at most $3(2c_5\log r_1 +
j+2)$.
There is an additional call to {\sc Update} in
Line~\ref{line:minchild_b}. As before, since this call is made on
behalf of a height-0 vertex, its cost is at most 2.

We have therefore established that the total amount of 
update debt incurred on behalf of vertices of height at most $j$  and 
paid off from accounts $\set{(R,x)\, :\, x \text{ an entry vertex of }R}$ is at
most $c_2 \sqrt{r_i} \beta_{i0}(2 + 3(2c_5\log r_1 + j + 2))$, which is
at most $3c_2 \sqrt{r_i} \beta_{i0} (3c_5\log r_1 + j)$. \qed
\end{proof}

\begin{lemma} \label{lem:total-update-debt}
The total update debt is at most
\begin{eqnarray}
2c_1c_2c_5 n r_1^{-1/2}\log r_1 & + & \\
9c_1c_2c_5n \sum_{i>0}  r_i^{-1/2} \beta_{i0} \log r_1 & +& \\
3c_1c_2n \sum_{i>0}  r_i^{-1/2} \beta_{i0} i & +& \\
27c_1c_2^2c_5n \sum_{j\geq 2}  r_j^{-1/2}  \sum_{1\leq i<j}
r_i^{1/2}\beta_{i0} \log r_1 & +& \\ 
9c_1c_2^2n \sum_{j\geq 2} r_j^{-1/2} \sum_{1\leq i<j}  r_i^{1/2}
\beta_{i0}  j & + &\\
27c_1c_2^2c_5^2n \sum_{j\geq 1}  r_j^{-1/2} \log^2 r_1 & + &\\
9c_1c_2^2c_5n \sum_{j\geq 1} r_j^{-1/2} j \log r_1 
\end{eqnarray}
\end{lemma}

\begin{proof}
To each unit of update debt, we associate two integers:
  $i$ is the height of the region $R$ such that the debt is paid off
  from an account $(R,x)$, and $j$ is the height of the vertex $v$ on
  whose behalf the debt was incurred.  If $i \geq j$, we refer to the
  debt as {\em type~1} debt, and if $i < j$, we refer to it as {\em
    type~2} debt.

First we bound the type-1 debt.  For each integer~$i>0$, there are at
most $c_1\frac{n}{r_i}$ regions $R$ of height~$i$.
By  Lemma~\ref{lem:update-debt-bound}, these contribute to the total type-1 debt
$\sum_{i>0} c_1 \frac{n}{r_i} 3c_2 \sqrt{r_i} \beta_{i0} (3c_5\log r_1 + i)$.

By Lemma~\ref{lem:num-0-regions}, there number of height~0 regions is 
$c_1c_2c_5\frac{n}{\sqrt r_1}\log r_1$, each with a single entry
vertex. An invocation on such a region makes at most one call to {\sc
  Update} that is on behalf of a height-0 vertex. By
Lemma~\ref{cor:cost-per-update-chain} the cost of such a call is
at most 2. Therefore, the contribution of height-0 regions to type-1
debt is   $2c_1c_2c_5\frac{n}{\sqrt r_1}\log r_1$.
We get that the total type-1 debt is at most
$$
2c_1c_2c_5\frac{n}{\sqrt r_1}\log r_1 + 
\sum_{i>0} 3c_1c_2 \frac{n}{\sqrt r_i} \beta_{i0} (3c_5\log r_1 + i).
$$

Now we bound the type-2 debt (i.e., $j>i$). 
For each integer~$j$ (since $j>i$, $j$ must be at least 1), the number of
regions of height~$j$ is at most $c_1 \frac{n}{r_j}$ and each has at
most $c_2 \sqrt{r_j}$ entry vertices, so the total number of vertices
of height~$j$ is at most $c_1 c_2 \frac{n}{\sqrt{r_j}}$. Consider each
such vertex $v$. 
We handle separately the
cases $i>0$ and $i=0$.
\begin{itemize}
\item For any $i>0$, every vertex (and in particular $v$) belongs to at most 3
  height-$i$ regions. 
Hence, by Lemma~\ref{lem:update-debt-bound}, the total
  type-2 debt for $i\geq 1$ is
$$\sum_{j\geq 2} c_1c_2 \frac{n}{\sqrt{r_j}} 3 \sum_{1\leq i<j}  3c_2 \sqrt{r_i} \beta_{i0} (3c_5\log r_1 + j).$$
\item We next consider the case $i=0$. For $j \geq 1$, there are $c_1c_2\frac{n}{\sqrt r_j}$ vertices
of height $j$. Each such vertex $v$  belongs to at most 3 height-1
regions, and hence to at most $3c_5 \log r_1$ height-0 regions.
Hence, by Lemma~\ref{lem:update-debt-bound}, the total type-2 debt for $i=0$ is 
$$ \sum_{j\geq 1} c_1c_2\frac{n}{\sqrt r_j} 3c_5 \log r_1 3c_2
(3c_5\log r_1 +j) = \\ 
\sum_{j\geq 1} 9c_1c_2^2c_5 \frac{n}{\sqrt r_j} \log r_1 
(3c_5\log r_1 +j) 
$$
\end{itemize}

The total update debt is thus as claimed in the lemma's statement. \qed
\end{proof} 

We are now prepared to bound the total running time of our algorithm to $O((n/\sqrt r)
\log^2 r)$. The following two lemmas bound the total process debt to 
$O(n/\sqrt r) \log^{c_q+1} r)$ and the total update debt to $O(n/\sqrt r) \log^{2} r)$. Using the naive RMQ data structure with $c_q=1$ we obtain the claimed running time.  

\begin{lemma}\label{lem:tot-process}
The total process debt is $O(\frac{n}{\sqrt r}
\log^{c_q+1} r)$.

\end{lemma}
\begin{proof}

Recall that $r_{i+1} = r_i^2$, so $r_i  = r_1^{2^{i-1}}$ and $\log
r_{i+1} = 2 \log r_i$. 
Assuming $r_1 \geq 2$, we have $\log r_{i+1}  > 2^{i} >
\left(\frac{4}{3}\right)^{i}$.
It follows that for any $i$
$$\left(\frac{4}{3}\right)^i \log r_{i+1} \leq 4 \log^2 r_i
$$

Observe that
\begin{eqnarray}
 r_i^{-1/2} \sum_{j \leq i} \beta_{ij} \log r_{j+1}
& = & \nonumber \\
r_i^{-1/2} \sum_{j \leq i}  \left(\frac{4}{3}\right)^{i-j}\frac{\log
  r_{i+1}}{\log r_{j+1}} \log r_{j+1}  & = & \nonumber \\
\left(\frac{4}{3}\right)^i r_i^{-1/2} \log r_{i+1} \sum_{j \leq i}
\left(\frac{3}{4}\right)^j & < & \nonumber \\
 4 r_i^{-1/2} \left(\frac{4}{3}\right)^i \log r_{i+1} & < & \nonumber \\
 16 r_i^{-1/2} \log^2 r_i \label{eq:proc1}
\end{eqnarray}

And that
\begin{eqnarray}
r_i^{-1/2} \beta_{i0}\log^{c_q} r_1 & = & \nonumber \\
r_i^{-1/2} \left(\frac{4}{3}\right)^i \frac{\log r_{i+1}}{\log
   r_1} \log^{c_q} r_1 & < & \nonumber \\
4 r_i^{-1/2} \log^2 r_i \log^{c_q-1} r_1 \label{eq:proc2}
\end{eqnarray}

Substituting~\eqref{eq:proc1} and~\eqref{eq:proc2} into
Corollary~\ref{cor:total-process-debt}, 
the total process debt is at most
\begin{eqnarray}
c_1 c_2 \sum_i \frac{n}{\sqrt{r_i}} \left( c_3\beta_{i0}\log^{c_q} r_1
  + \sum_{j \leq i} 4\beta_{ij} \log r_{j+1}\right) & < & \nonumber \\
 c_1c_2n \sum_i \left(  4c_3 r_i^{-1/2} \log^2 r_i \log^{c_q-1} r_1 \  
+ 64 r_i^{-1/2} \log^2 r_i \right) \nonumber
\end{eqnarray}

Since the $r_i$'s increase doubly exponentially, 
we have that $ \sum_{i>0}  r_i^{-1/2} \log^2 r_i  = O(r_1^{-1/2} \log^2 r_1 ) $
and so the  total process debt is bounded by  $O(\frac{n}{\sqrt r}
\log^{c_q+1} r)$.  \qed
\end{proof}

\begin{lemma}
The total update debt is
$O(\frac{n}{\sqrt r_1}\log^2 r_1)$.
\end{lemma}
\begin{proof}
We need to bound every term in the bound stated in
Lemma~\ref{lem:total-update-debt}. 
\begin{itemize}

\item The first term is clearly $O(\frac{n}{\sqrt r_1}\log r_1)$.
\item By~\eqref{eq:proc2}, the second term is  
$$36c_1c_2c_5n \sum_{i>0}  r_i^{-1/2} \log^2 r_i. $$
Recall that $r_{i+1} = r_i^2$, so $r_i  = r_1^{2^{i-1}}$.
Since the $r_i$'s
increase doubly exponentially, we have that $\sum_{i>0}  r_i^{-1/2} \log^2 r_i = O(r_1^{-1/2} \log^2 r_1 ) $
and so the term is bounded by $O(\frac{n}{\sqrt r_1}\log^2 r_1)$.

\item Similarly to the derivation of~\eqref{eq:proc2}, 
\begin{eqnarray}
r_i^{-1/2} \beta_{i0} i  & = & \nonumber \\
r_i^{-1/2} \left(\frac{4}{3}\right)^i \frac{\log r_{i+1}}{\log
   r_1} i & < & \nonumber \\
4 i r_i^{-1/2} \log^2 r_i   \log^{-1} r_1\label{eq:upd3}
\end{eqnarray}
Substituting~\eqref{eq:upd3} into the third term we get 
$3c_1c_2n \sum_{i>0}  i r_i^{-1/2} \log^2 r_i  \log^{-1} r_1$. 
Since the $r_i$'s increase doubly exponentially, 
we have that $ \sum_{i>0}  i r_i^{-1/2} \log^2 r_i  = O(r_1^{-1/2} \log^2 r_1 ) $
and so the  term is bounded by $O(\frac{n}{\sqrt r_1}\log r_1)$.

\item
We need to bound the fourth term, which is 
$$27c_1c_2^2c_5n \sum_{j\geq 2}  r_j^{-1/2}  \sum_{1\leq i<j}
r_i^{1/2}\beta_{i0} \log r_1. $$
We first focus on the second sum.
\begin{eqnarray}
\sum_{1\leq i<j} r_i^{1/2}\beta_{i0} & = & \nonumber \\
 \sum_{1\leq i<j} r_i^{1/2} \left(\frac{4}{3}\right)^i \frac{\log r_{i+1}}{\log
   r_1} & < & \nonumber \\
4 \sum_{1\leq i<j} r_i^{1/2} \log^2 r_i \log^{-1}r_1& < & \nonumber \\
c_6 r_{j-1}^{1/2} \log^2 r_{j-1} \log^{-1}r_1& & \label{eq:upd4}
\end{eqnarray}
In the last step we again used the exponential increase of the
$r_i$'s. 
The fourth term is therefore bounded by 
\begin{eqnarray}
27c_1c_2^2c_5n \sum_{j\geq 2}  r_j^{-1/2} c_6  r_{j-1}^{1/2} \log^2
r_{j-1} & = & \nonumber \\
27c_1c_2^2c_5c_6n \sum_{j\geq 2}  r_j^{-1/2} r_j^{1/4} \left(
  \frac{1}{2} \log r_j \right)^2 & = & \nonumber \\
\frac{27}{4}c_1c_2^2c_5c_6n \sum_{j\geq 2}  r_j^{-1/4} \log^2 r_j  & & \nonumber 
\end{eqnarray}
Again, this sum is dominated by the leading term, so it is 
$O(n r_2^{-1/4} \log^2 r_2) = O(\frac{n}{\sqrt r_1}\log^2 r_1) $.

\item The fifth term we need to bound is 
$$9c_1c_2^2n \sum_{j\geq 2} r_j^{-1/2} \sum_{1\leq i<j}  r_i^{1/2}
\beta_{i0}  j.$$
Using~\eqref{eq:upd4}, it is bounded by
\begin{eqnarray}
9c_1c_2^2n \sum_{j\geq 2} r_j^{-1/2} j c_6 r_{j-1}^{1/2} \log^2 r_{j-1}
\log^{-1}r_1 & < & \nonumber \\
\frac{9}{4}c_1c_2^2c_6n \log^{-1}r_1 \sum_{j\geq 2} j r_j^{-1/4}  \log^2 r_j
 & < & \nonumber \\
O(n r_2^{-1/4} \log^2 r_2 \log^{-1}r_1) & = & \nonumber \\
O(\frac{n}{\sqrt r_1} \log r_1) & & \nonumber 
\end{eqnarray}

\item In the sixth term $27c_1c_2^2c_5^2n \sum_{j\geq 1}  r_j^{-1/2}
  \log^2 r_1$, the sum is also bounded by its leading term, so we get
$O(\frac{n}{\sqrt r_1} \log^2 r_1)$.

\item Similarly, in the seventh term $9c_1c_2^2c_5n \sum_{j\geq 1}
  r_j^{-1/2} j \log r_1 $, the sum is bounded by its leading term so we get
  $O(\frac{n}{\sqrt r_1} \log r_1)$.\qed
\end{itemize} 
\end{proof}

Combining the total $O(\frac{n}{\sqrt r} \log^{c_q+1} r)$ process
debt, and the total $O(\frac{n}{\sqrt r} \log^2 r)$ update debt, we
get that the total running time of the algorithm is 
$O(\frac{n}{\sqrt r} \log^{c_q+1} r)$. 

By Lemma~\ref{lem:HKRS-RMQ}, The parameter $c_q$ is either 1 or 3,
thus proving Theorem~\ref{thm:analysis}

\bibliographystyle{plain}

\begin{thebibliography}{10}

\bibitem{BH13}
G.~Borradaile and A.~Harutyunyan.
\newblock Maximum st-flow in directed planar graphs via shortest paths.
\newblock In {\em Combinatorial Algorithms}, pages 423--427. Springer, 2013.

\bibitem{borradaile-klein-09}
G.~Borradaile and P.~Klein.
\newblock An ${O}(n \log n)$ algorithm for maximum {\em st}-flow in a directed
  planar graph.
\newblock {\em Journal of the {ACM}}, 56(2):1--30, 2009.

\bibitem{BKMNWN11}
G.~Borradaile, P.~N. Klein, S.~Mozes, Y.~Nussbaum, and C.~Wulff-Nilsen.
\newblock Multiple-source multiple-sink maximum flow in directed planar graphs
  in near-linear time.
\newblock In {\em Proceedings of the 52nd Annual Symposium on Foundations of
  Computer Science {(FOCS)}}, pages 170--179, 2011.

\bibitem{min-cut}
G.~Borradaile, P.~Sankowski, and C.~Wulff-Nilsen.
\newblock Min $st$-cut oracle for planar graphs with near-linear preprocessing
  time.
\newblock In {\em Proceedings of the 51st Annual Symposium on Foundations of
  Computer Science {(FOCS)}}, pages 601--610, 2010.

\bibitem{Cabello}
S.~Cabello.
\newblock Many distances in planar graphs.
\newblock {\em Algorithmica}, 62(1-2):361--381, 2012.
\newblock Preliminary version in SODA 2006.

\bibitem{CabelloCE13}
S.~Cabello, E.W. Chambers, and J.~Erickson.
\newblock Multiple-source shortest paths in embedded graphs.
\newblock {\em SIAM Journal on Computing}, 42(4):1542--1571, 2013.

\bibitem{ChFaNa04}
P.~Chalermsook, J.~Fakcharoenphol, and D.~Nanongkai.
\newblock A deterministic near-linear time algorithm for finding minimum cuts
  in planar graphs.
\newblock In {\em Proceedings of the 14th annual Symposium On Discrete
  Algorithms ({SODA})}, pages 828--829, 2004.

\bibitem{ChangL13}
H.-C. Chang and H.-I Lu.
\newblock Computing the girth of a planar graph in linear time.
\newblock {\em SIAM Journal on Computing}, 42(3):1077--1094, 2013.

\bibitem{CGfrac}
B.~Chazelle and L.~J. Guibas.
\newblock Fractional cascading: {I}. {A data structuring technique}.
\newblock {\em Algorithmica}, 1:133--162, 1986.

\bibitem{EisenstatK13}
D.~Eisenstat and P.~N. Klein.
\newblock Linear-time algorithms for max flow and multiple-source shortest
  paths in unit-weight planar graphs.
\newblock In {\em Proceedings of the 45th {ACM} Symposium on Theory of
  Computing, ({STOC})}, pages 735--744, 2013.

\bibitem{Erickson2010}
J.~Erickson.
\newblock Maximum flows and parametric shortest paths in planar graphs.
\newblock In {\em Proceedings of the 21st annual Symposium On Discrete
  Algorithms ({SODA})}, pages 794--804, 2010.

\bibitem{FR06}
J.~Fakcharoenphol and S.~Rao.
\newblock Planar graphs, negative weight edges, shortest paths, and near linear
  time.
\newblock {\em Journal of Computer and System Sciences}, 72:868--889, 2006.

\bibitem{Frederickson}
G.~N. Frederickson.
\newblock Fast algorithms for shortest paths in planar graphs, with
  applications.
\newblock {\em SIAM Journal on Computing}, 16(6):1004--1022, 1987.

\bibitem{FT87}
M.~L. Fredman and R.~E. Tarjan.
\newblock Fibonacci heaps and their uses in improved network optimization
  algorithms.
\newblock {\em Journal of the ACM}, 34(3):596--615, 1987.

\bibitem{GMW}
P.~Gawrychowski, S.~Mozes, and O.~Weimann.
\newblock Improved submatrix maximum queries in {M}onge matrices.
\newblock arXiv:1307.2313, 2013.

\bibitem{Hassin}
R.~Hassin.
\newblock Maximum flow in $(s,t)$ planar networks.
\newblock {\em Inform. Process. Lett.}, 13(3):107, 1981.

\bibitem{HKRS97}
M.~R. Henzinger, P.~N. Klein, S.~Rao, and S.~Subramanian.
\newblock Faster shortest-path algorithms for planar graphs.
\newblock {\em Journal of Computer and System Sciences}, 55(1):3--23, 1997.

\bibitem{IS79}
A.~Itai and Y.~Shiloach.
\newblock Maximum flow in planar networks.
\newblock {\em SIAM Journal on Computing}, 8(2):135--150, 1979.

\bibitem{MinstCutSTOC}
G.~F. Italiano, Y.~Nussbaum, P.~Sankowski, and C.~Wulff-Nilsen.
\newblock Improved algorithms for min cut and max flow in undirected planar
  graphs.
\newblock In {\em Proceedings of the 43rd {ACM} Symposium on Theory of
  Computing, ({STOC})}, pages 313--322, 2011.

\bibitem{JV83}
D.B. Johnson and S.~M. Venkatesan.
\newblock {Partition of planar flow networks}.
\newblock In {\em Proceedings of the 24th Annual Symposium on Foundations of
  Computer Science {(FOCS)}}, pages 259--264, 1983.

\bibitem{KaplanMNS12}
H.~Kaplan, S.~Mozes, Y.~Nussbaum, and M.~Sharir.
\newblock Submatrix maximum queries in {M}onge matrices and {M}onge partial
  matrices, and their applications.
\newblock In {\em Proceedings of the 23rd Annual {ACM}-{SIAM} Symposium On
  Discrete Mathematics ({SODA})}, pages 338--355, 2012.

\bibitem{KN11a}
H.~Kaplan and Y.~Nussbaum.
\newblock Maximum flow in directed planar graphs with vertex capacities.
\newblock {\em Algorithmica}, 61(1):174--189, 2011.

\bibitem{KaplanN11}
H.~Kaplan and Y.~Nussbaum.
\newblock Minimum $s$-$t$ cut in undirected planar graphs when the source and
  the sink are close.
\newblock In {\em Proceedings of the 28th International Symposium on
  Theoretical Aspects of Computer Science (STACS)}, pages 117--128, 2011.

\bibitem{KNK93}
S.~Khuller, J.~Naor, and P.~Klein.
\newblock The lattice structure of flow in planar graphs.
\newblock {\em SIAM J.\ Discret.\ Math.}, 6(3):477--490, 1993.

\bibitem{Klein05}
P.~N. Klein.
\newblock Multiple-source shortest paths in planar graphs.
\newblock In {\em Proceedings of the 16th Annual {ACM}-{SIAM} Symposium On
  Discrete Mathematics ({SODA})}, pages 146--155, 2005.

\bibitem{Book}
P.~N. Klein and S.~Mozes.
\newblock {\em Optimization Algorithms for Planar Graphs}.
\newblock http://planarity.org/, 2014.

\bibitem{ShayrDivision}
P.~N. Klein, S.~Mozes, and C.~Sommer.
\newblock Structured recursive separator decompositions for planar graphs in
  linear time.
\newblock In {\em Proceedings of the 44th {ACM} Symposium on Theory of
  Computing, ({STOC})}, pages 505--514, 2012.

\bibitem{KMW10}
P.~N. Klein, S.~Mozes, and O.~Weimann.
\newblock Shortest paths in directed planar graphs with negative lengths: A
  linear-space ${O}(n \log^2 n)$-time algorithm.
\newblock {\em ACM Transactions on Algorithms}, 6(2):1--18, 2010.
\newblock Preliminary version in SODA 2009.

\bibitem{LackiSankowski}
J.~\L\k{a}cki and P.~Sankowski.
\newblock Min-cuts and shortest cycles in planar graphs in ${O}(n \log \log n)$
  time.
\newblock In {\em Proceedings of the 19th annual European Symposium on
  Algorithms (ESA)}, pages 155--166, 2011.

\bibitem{MST}
T.~Matsui.
\newblock The minimum spanning tree problem on a planar graph.
\newblock {\em Discrete Applied Mathematics}, 58(1):91--94, 1995.

\bibitem{MozesSommer}
S.~Mozes and C.~Sommer.
\newblock Exact distance oracles for planar graphs.
\newblock In {\em Proceedings of the 23rd {ACM}-SIAM Symposium On Discrete
  Algorithms ({SODA})}, pages 571--582, 2012.

\bibitem{ShayWulf}
S.~Mozes and C.~Wulff-Nilsen.
\newblock Shortest paths in planar graphs with real lengths in ${O}(n \log^2n /
  \log\log n)$.
\newblock In {\em Proceedings of the 18th annual European Symposium on
  Algorithms (ESA)}, pages 206--217, 2010.

\bibitem{Nussbaum11}
Y.~Nussbaum.
\newblock Improved distance queries in planar graphs.
\newblock In {\em Proceedings of the 14th International Workshop on Algorithms
  and Data Structures {(WADS)}}, pages 642--653, 2011.

\bibitem{reif83}
J.~Reif.
\newblock Minimum $s$-$t$ cut of a planar undirected network in ${O}(n \log^2
  n)$ time.
\newblock {\em SIAM Journal on Computing}, 12:71--81, 1983.

\bibitem{Subramanian95}
S.~Subramanian.
\newblock {\em {P}arallel and {D}ynamic {S}hortest-{P}ath {A}lgorithms for {S
  }parse {G}raphs}.
\newblock PhD thesis, Brown University, 1995.
\newblock Available as Brown University Computer Science Technical Report
  CS-95-04.

\bibitem{TM09}
S.~Tazari and M.~M{\"u}ller-Hannemann.
\newblock Shortest paths in linear time on minor-closed graph classes, with an
  application to steiner tree approximation.
\newblock {\em Discrete Applied Mathematics}, 157(4):673--684, 2009.

\bibitem{Willard1983}
D.~E. Willard.
\newblock Log-logarithmic worst-case range queries are possible in space
  $\theta({N})$.
\newblock {\em Inf. Process. Lett.}, 17(2):81--84, 1983.

\end{thebibliography}

\end{document}